\documentclass[journal, twocolumn, a4paper]{IEEEtran}

\usepackage{amsthm, amsfonts} 
\usepackage{extarrows}         
\usepackage{algorithm}        
\usepackage{algpseudocode}    
\usepackage{graphicx}         
\usepackage{epstopdf}         
\usepackage{multirow}          
\usepackage{multicol}          
\usepackage{color}             
\usepackage{cite}             
\usepackage{url}              
\usepackage{setspace}         
\usepackage{caption}          
\usepackage{subcaption}       
\usepackage{makecell}         
\usepackage{siunitx}          
\usepackage{mypkg}            

\makeatletter

\newcommand{\Rmnum}[1]{\expandafter\@slowromancap\romannumeral #1@}
\makeatother

\renewcommand{\algorithmicrequire}{ \textbf{Input:}}
\renewcommand{\algorithmicensure}{ \textbf{Output:}}

\theoremstyle{plain}

\newtheorem{proposition}{Proposition}

\begin{document}
\include{com.tex}

\title{Complexity-Scalable Near-Optimal Transceiver Design for Massive MIMO-BICM Systems}
\author{Jie Yang, Wanchen Hu, Yi Jiang,~\IEEEmembership{Member,~IEEE}, Shuangyang Li,~\IEEEmembership{Member,~IEEE}, \\Xin Wang,~\IEEEmembership{Fellow,~IEEE}, Derrick Wing Kwan Ng,~\IEEEmembership{Fellow,~IEEE}, and Giuseppe Caire,~\IEEEmembership{Fellow,~IEEE}
}
\maketitle
 \blfootnote{
J. Yang, W. Hu, Y. Jiang, and X. Wang are with Key Laboratory for Information Science of Electromagnetic Waves (MoE), School of Information Science and Technology, Fudan University, Shanghai, China (E-mails: \{yangjie23, huwc23\}@m.fudan.edu.cn, \{yijiang,  xwang11\}@fudan.edu.cn.) \\
\indent D. W. K. Ng is with the School of Electrical Engineering and Telecommunications, University of New South
Wales, Sydney, NSW 2052, Australia (email: w.k.ng@unsw.edu.au),\\
\indent S. Li and G. Caire
are with Communications and Information Theory Group (CommIT), Technische Universitat Berlin, 10587 Berlin, Germany (e-mails:\{
shuangyang.li, caire\}@tu-berlin.de.)
}
\begin{abstract}
Future wireless networks are envisioned to employ multiple-input multiple-output (MIMO) transmissions with large array sizes, and therefore, the adoption of complexity-scalable transceiver becomes important. In this paper, we propose a novel complexity-scalable transceiver design for MIMO systems exploiting bit-interleaved coded modulation (termed MIMO-BICM systems). The proposed scheme leverages the channel bidiagonalization decomposition (CBD), based on which an optimization framework for the precoder and post-processor is developed for maximizing the mutual information (MI) with finite-alphabet inputs. Particularly, we unveil that the desired precoder and post-processor behave distinctively with respect to the operating signal-to-noise ratio (SNR), where the equivalent channel condition number (ECCN) serves as an effective indicator for the overall achievable rate performance. Specifically, at low SNRs, diagonal transmission with a large ECCN is advantageous, while at high SNRs, uniform subchannel gains with a small ECCN are preferred. This allows us to further propose a low-complexity generalized parallel CBD design (GP-CBD) based on Givens rotation according to a well-approximated closed-form performance metric on the achievable rates that takes into account the insights from the ECCN.
Numerical results validate the superior performance of the proposed scheme in terms of achievable rate and bit error rate (BER), compared to state-of-the-art designs across various modulation and coding schemes (MCSs).

\end{abstract}

\section{Introduction}
Multiple-input multiple-output bit-interleaved coded modulation (MIMO-BICM) is a cornerstone of modern wireless communications and has been widely adopted in cellular networks and wireless local area networks (WLANs) \cite{BICM,wang2020joint,bjornson2017massive}. Furthermore, massive MIMO with BICM has emerged as a key component in the fifth-generation (5G) and beyond communication systems. As the number of antennas increases, receiver complexity escalates dramatically. In massive MIMO systems, maximum likelihood (ML) and maximum a posteriori (MAP) detection require an exhaustive search over all possible transmitted symbols to minimize detection and decoding errors. However, with large antenna arrays and high-order modulation, the computational complexity of these optimal algorithms becomes prohibitive, rendering them impractical for real-time transmission. Conversely, linear detection techniques, such as zero-forcing (ZF) and linear minimum mean square error (MMSE), along with nonlinear decision feedback equalization (DFE), are widely considered due to their lower computational cost, albeit at the expense of degraded error performance. Over the past few decades, extensive research has been devoted to achieving near-optimal performance in MIMO systems given the channel state information (CSI). These efforts are broadly categorized into two approaches: (i) mutual information (MI)-based designs and (ii) bit error rate (BER)-based designs. Gradient descent-based methods \cite{9343768, jing2021linear, 9693344, 10534211} have been widely explored for precoder optimization to maximize MI in MIMO systems. On the BER front, various suboptimal nonlinear detection algorithms have been proposed, including sphere decoding (SD) \cite{SD, barbero2008fixing}, Bayesian inference-based techniques \cite{takahashi2021low, 10480363}, and neural network-assisted methods \cite{zheng2025low, OLTD, PMAP}. While these approaches improve performance, they generally incur significantly higher complexity than MMSE-based techniques, rendering them impractical for real-time applications.


With CSI available at both the transmitter and receiver, holistic transceiver optimization can fully exploit the transmission capacity and reliability of MIMO systems while maintaining manageable complexity. It is well established that singular value decomposition (SVD)-based designs when paired with MMSE reception, achieve the channel capacity with Gaussian inputs in MIMO systems \cite{telatar1999capacity, perovic2021achievable}. However, the performance of SVD-based schemes is constrained by the weakest subchannel in the high signal-to-noise ratio (SNR) regime, especially for discrete modulation inputs with finite constellation sizes \cite{GMD}. This limitation becomes more pronounced in massive MIMO systems, where the channel becomes significantly ill-conditioned, further exacerbating performance degradation. To mitigate the impact of the worst subchannel, geometric mean decomposition (GMD)-based methods have been proposed \cite{GMD, UCD}. Unfortunately, DFE in GMD-based schemes usually leads to a significant error propagation. Besides, Mohammed \textit{et al.} \cite{mohammed2011mimo} proposed an X-code structure to approach ML-decoding performance. However, this efficient X-code suffers from the performance degradation under ill-conditioned channels \cite{mohammed2011precoding}. More recently, a channel bidiagonal (CBD) structure-based on matrix decomposition for MIMO systems was introduced in \cite{CBDbib}. This structure supports ML reception with the same complexity as MMSE for MIMO-BICM systems. Notably, the authors in \cite{maleki2024precoding} demonstrated that the precoding strategies in MIMO-BICM systems significantly differ from those in traditional MIMO systems, and emphasized the necessity of transceiver optimization for MIMO-BICM systems. The authors further validated this claim through experimental results. This distinction arises because transmission strategies in MIMO-BICM systems must consider both channel conditions and the discrete bit-level modulation, whereas traditional MIMO systems focus primarily on maximizing power or energy efficiency.

Despite the importance of transceiver optimization in MIMO-BICM systems, research in this area is limited. Notable studies, such as \cite{maleki2024precoding} proposed a gradient descent method to maximize the MI for each channel realization in MIMO-BICM systems at various SNRs. This approach suffers from prohibitively high computational complexity. To address this issue, the authors also explored a low-complexity method by Monte Carlo trials to establish a connection between transceiver design and channel coding rate, at the expense of performance degradation. In addition, a block bidiagonal scheme was also proposed in~\cite{PL-CBD}, which leverages a hyperparameter to adjust the size of the CBD, without highlighting the theoretical analysis of the achievable rate performance. Indeed, the studies on the theoretical analysis for MIMO-BICM are relatively scarce in the current literature. The most relevant theoretical attempt to evaluate the performance of MIMO-BICM systems appeared in \cite{fertl2011performance}, where the authors provided extensive experimental evidence suggesting that no universal transceiver design can produce superior MI and BER performance from low to high SNR regimes. 

Due to the lack of theoretical analysis in MIMO-BICM systems, transceiver designs have relied primarily on heuristic optimization methods to improve MI and BER performance. However, these approaches often result in high complexity or performance degradation under ill-conditioned channels \cite{CBDbib,maleki2024precoding,PL-CBD,fertl2011performance}. Motivated by these limitations, {this paper proposes a comprehensive theoretical framework for non-iterative\footnote{A comprehensive performance assessment of soft-in soft-out demodulators in iterative BICM receivers entails a fundamentally different analytical framework and is therefore beyond the scope of this paper.}} MIMO-BICM systems. This framework, in turn, enables novel and efficient optimization strategies for general MIMO-BICM systems given the CSI.

{In particular, we adopt an information-theoretic approach to evaluate the MI of the MIMO-BICM transceiver design with finite-alphabet inputs. We restrict ourselves in considering the case where the instantaneous CSI is available perfectly at the transmitter side, such that the achievable rate and BER averaged across sufficient number of channel realizations (in the ergodic sense) becomes relevant performance metric.} A concise expression is derived for the achievable rate of the MIMO-BICM system as a function of the \textit{a posteriori} log-likelihood ratio (LLR), based on which, further approximations on the achievable rate are also presented in order to facilitate the efficient precoding design per channel realization.
The main contributions of this paper are summarized as follows:
\begin{itemize}

\item We leverage MI analysis to derive a concise expression for the achievable rate of MIMO-BICM systems with finite-alphabet inputs. To enhance analytical tractability, we establish a lower bound for the achievable rate function in both low and high SNR regimes and demonstrate its asymptotic optimality. Building on these theoretical insights, we identify the opposite characteristics of the transceiver design at low and high SNRs, which can be characterized by the equivalent channel condition
number (ECCN). Specifically, we show that, in terms of maximizing the achievable rate, a channel matrix with a larger ECCN is advantageous in the low SNR regime, while a channel matrix with a smaller ECCN is preferred in the high SNR regime.


\item Guided by the derived MI analysis, we propose a novel  generalized parallel CBD (GP-CBD) design. This design strategically leverages eigen-subspaces permutation and rotations, along with partial block bidiagonalization for precoder and post-processor optimization. Due to this block structure, the proposed GP-CBD scheme enables an efficient implementation with a scalable complexity to the number of antennas, which is particularly suitable for massive MIMO applications. More importantly, the GP-CBD scheme is highly adaptive to the preferred ECCNs with respect to different SNRs. This adaptability allows the GP-CBD scheme to achieve promising performance in terms of achievable rate and BER in varying SNR regimes. 

\item Extensive numerical results are provided based on the independent and identical distributed (i.i.d.) Rayleigh fading channel and the clustered delay line (CDL) channel model specified
in 3GPP TR 38.901~\cite{38901}. Particularly, our numerical results demonstrate near-optimal performance of the proposed scheme with an overall complexity linear to the size of the constellation and number of transmit antennas, while quadratic to the number of receive antennas.
\end{itemize}

\begin{figure*}[ht!]
\centering
\includegraphics[width=6.7in]{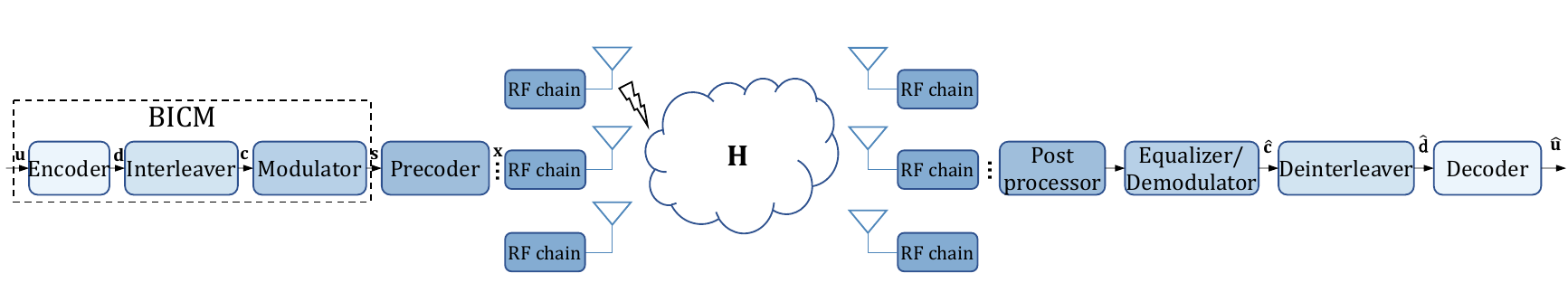}
\caption{The considered MIMO-BICM system.}
\label{BICM}
\end{figure*}

The remainder of this paper is organized as follows. Section \ref{System-Model} describes the MIMO-BICM system model and briefly reviews the conventional transceiver design principles. In Section \ref{Performance and Comp}, we derive a succinct expression of the achievable rate for MIMO-BICM systems and provide a guideline for the development of the desired CBD based on MI analysis. In Section \ref{SEC:GP-CBD}, we propose a new transceiver design called GP-CBD to achieve high performance with low complexity for massive MIMO-BICM systems. Numerical results are provided in Section \ref{SEC:results}. Section \ref{SEC:conclusion} concludes this paper.

\emph{Notations:} Boldface uppercase (lowercase) letters denote matrices, e.g., $\mathbf{A}$, (column vectors, e.g., $\mathbf{a}$), and italics denote scalars, e.g., $a$. The $(i, j)$-th entry of $\mathbf{A}$ is denoted by $A_{i,j}$. The columns of an $M\times N$ matrix $\mathbf{A}$ are denoted by $\mathbf{a}_1,\mathbf{a}_2,\cdots,\mathbf{a}_{N}$; $\mathbb {C}^{M\times N}$ and $\mathbb {R}^{M\times N}$ denote the sets of $M\times N$ complex and real matrices, respectively. The superscripts $(\cdot)^T$ and $(\cdot)^H$ stand for the transport and conjugate transport operators, respectively; $\rm Tr(\mathbf{A})$ represents the trace of matrix $\mathbf{A}$; $\mathbf{I}_N$ represents an $N \times N$ identity matrix. The operator $\mathbb{E}[\cdot]$ denotes the statistical expectation; $\log(\cdot)$ stands for the base natural logarithm; $\| \cdot \|$ stands for the $2$-norm operation; $ \lfloor \cdot \rfloor$ represents the downward rounding operation; $p(\cdot)$ stands for the probability density function (PDF); and $\mathcal{H}(\cdot)$ stands for the entropy function; $\Re\{\cdot\}$ and $\Im\{\cdot\}$ represent the in-phase (or real) and quadrature (or imaginary) parts of data, respectively. $\rm blkdiag(\cdot)$ and $\rm diag(\cdot)$ stand for the block diagonal matrix and diagonal matrix, respectively. 

\section{System Model And Preliminaries} \label{System-Model}
This section begins with an introduction to the MIMO-BICM system model. We then present the conventional SVD-based transceiver and discuss its inherent challenges in massive MIMO systems. Additionally, a brief overview of the CBD scheme is provided.


\subsection{MIMO-BICM System Model}


We consider a point-to-point MIMO-BICM system with $N_t$ transmit antennas and $N_r$ receive antennas, as depicted in Fig. \ref{BICM}, where a sequence of information bits $u_i, i =1, 2, \cdots $ is transmitted. The bit sequence is first encoded by a forward error correction (FEC) code, resulting in a sequence of coded bits $d_k, k = 1, 2, \cdots$. These coded bits are then interleaved into bits $c_k, k = 1, 2, \cdots$. Next, the interleaved bits are grouped into subsequences of $K$ bits and mapped to constellation points by the modulator. Finally, the modulated symbols are multiplexed into the transmitted signal $\mathbf{s}\in {\mathbb C}^{N_s \times 1}$. The sequence contains $N_s$ information symbols, where each entry is modulated by $Q_m$ interleaved bits and taken from an $M$-QAM constellation set $\mathcal{M}$. Here, we assume $ N_s =\min\{N_t, N_r\}$ for the ease of exposition.  
Then, we have $K=Q_mN_s$ and $\mathbb E(\mathbf{s} \mathbf{s}^H)=\mathbf{I}_{N_s}$. 
For convenience, we define $\mathcal{M}^{c_k}$ as the set of \textit{symbols} given $c_k$, which contains symbols from $\mathcal{M}$ that correspond to $c_k$ based on the underlying mapping rule. Furthermore, $\mathcal{X}^{c_k}$ is defined as the set of symbol vectors given $c_k$, including all possible symbol vectors containing $c_k$ as the $k$-th bit.

Given the precoder $\mathbf{F} \in {\mathbb C}^{N_t\times N_s}$, the receiver obtains
\begin{equation}\label{eq1}
\Bar{\mathbf{y}} = \mathbf{HFs} + \Bar{\mathbf{z}},
\end{equation}
 where $\mathbf{H} \in \mathbb C^{N_r\times N_t}$ is the channel matrix,  $\Bar{\mathbf{z}}$ is the zero-mean circularly symmetric white Gaussian noise vector with $\mathbb E(\Bar{\mathbf{z}} \Bar{\mathbf{z}}^H)=\sigma^2_z\mathbf{I}_{N_r}$, and $\sigma^2_z$ is the noise power per receive antenna. 

Assuming that the CSI is available at the receiver, the sphere decoding (SD) algorithm \cite{SD} is proposed to address the following problem
\begin{equation}
\min_{\mathbf{s} \in {\cal X}} \| \Bar{\mathbf{y}} - \mathbf{HFs} \|^2, \label{eq2} 
\end{equation}
without exhaustively searching the entire solution space, where $\mathcal{X}$ denotes the set of the transmitted symbol vectors. After the post-processing and equalization stages, the soft information $\hat{c}_k,k=1,2,\cdots,K$, is obtained and subsequently fed into the decoder to obtain $\hat{u}$. Nevertheless, even for a MIMO system with moderate array size, e.g., $N_s=8$ and $M=64$, the SD algorithm becomes intractable for practical implementation, not to mention its soft version \cite{SoftSD,SOftSD1}, which prohibits its direct application in future wireless systems equipped with massive antennas. In fact, most practical MIMO systems consider certain reduced-complexity detection schemes together with well-designed precoding and post-processing modules for data detection \cite{GMD,PL-CBD}. Among these schemes, channel decomposition methods, such as the SVD-based transceiver design is the popular choice.

\subsection{Channel decomposition Preliminaries and the Challenge for SVD in Massive MIMO}
For the channel matrix $\mathbf{H}$, its generalized triangular decomposition is given by
\begin{equation}\label{eq:QRP}
    \mathbf{H} = \mathbf{QR P}^H,
\end{equation}
where $\mathbf{Q} \in \mathbb C^{N_r\times N_r} $ and $\mathbf{P}\in \mathbb C^{N_t\times N_s}$ are unitary matrices, and $\mathbf{R} \in \mathbb R^{N_r \times N_t}$ is an upper triangular matrix. 
The upper-triangular structure enables efficient implementation of DFE \cite{GMD,UCD}, which is commonly applied for data detection. In particular, the so-called ``equivalent channel condition number (ECCN)'' is a parameter that characterizes the error performance of DFE, which is defined as 
\begin{equation}\label{eq:ECCN}
    {\rm cond(\mathbf{R})} \triangleq \frac{\max R_{i,i}}{\min R_{i,i}}.
\end{equation}
It has been shown that the diagonal entries of $\mathbf{R}$ significantly influence the minimum Euclidean distance (MED) among possible received signals in the noiseless regime with the given channel matrix \cite{EqualQR}, directly impacting the system performance. In particular,  $\mathbf{R}$ can transform to a diagonal matrix. Then, (\ref{eq:QRP}) reduces to the SVD, given by
\begin{equation}\label{eq:svd}
    \mathbf{H} = \mathbf{U\Lambda V}^H,
\end{equation} 
where $\mathbf{U}\in \mathbb C^{N_r\times N_r} $ and $\mathbf{V}\in \mathbb C^{N_t\times N_s}$ are unitary matrices, obtained by proper unitary transformation from $\mathbf{Q}$ and $\mathbf{P}$, respectively; $\mathbf{\Lambda} = {\rm diag}\{\lambda_1,\cdots,\lambda_{N_s}\}$\footnote{Without loss the generality, we assume here that the singular values are in descending order, i.e., $\lambda_1 \geq \lambda_2 \geq \cdots \geq \lambda_{N_s}$.}. Moreover, (\ref{eq:QRP}) and (\ref{eq:svd}) should satisfy the majorization theory \cite{horn}
\begin{equation}\label{eq:horn1}
\begin{aligned}
        &\prod_{i=1}^{N_s} R_{i,i} = \prod_{i=1}^{N_s} \lambda_{i,i},
        &\prod_{i=1}^{n} R_{i,i} \leq \prod_{i=1}^{n} \lambda_{i,i}, 1\leq n\leq N_s.
\end{aligned}
\end{equation}
Consequently, SVD transforms the MIMO channel into $N_s$ eigen-subchannels. However, the presence of small $\lambda_n$ significantly degrades the output SNR, which can be characterized by the large ECCN ${\rm cond(\mathbf{\Lambda})} = \frac{\lambda_{1}}{\lambda_{N_s}}$. For MIMO systems operating in rich-scattered scenarios, the performance loss of SVD remains manageable due to the well-conditioned nature of $\mathbf{H}$. However, in massive MIMO systems with high SNRs, the performance of SVD-based transceivers is generally limited due to the significant path loss and potential line-of-sight dominance in the far-field transmission \cite{liu2024leveraging}. 
As we will show in Section \ref{SUBSEC:GP-CBD}, the performance of the SVD-based scheme deteriorates in high SNR regimes as ECCN increases, which is a common case in point-to-point MIMO transmissions with massive arrays \cite{9133524}.
As such, we are motivated to propose a complexity-scalable and high performance transceiver design for general MIMO-BICM systems based on the CBD scheme described below.

\subsection{CBD Transceivers}\label{SUB:BCJR}
The bidiagonalization of $\mathbf{H}$ can be obtained via a series of Householder transformations \cite{CBDbib} as
\begin{equation}\label{eq:CBD}
  \mathbf{H} = \mathbf{QBP}^H,
\end{equation}
where $\mathbf{Q} \in \mathbb C^{N_r\times N_r} $ and $\mathbf{P}\in \mathbb C^{N_t\times N_s}$ are unitary matrices, and $\mathbf{B} \in \mathbb R^{N_r \times N_t}$ is a \textit{real-valued} upper bidiagonal matrix. It is important to point out that $\mathbf{\Sigma}$ in (\ref{eq:svd}) is a special case of $\mathbf{B}$ with the subdiagonal entries being zeros. When CSI are available at both transmitter and receiver, the transmitter can employ $\mathbf{F} = \mathbf{P}$ as the precoder and the receiver adopts $\mathbf{Q}^H$ as the post-processor. Consequently, the original channel in (\ref{eq1}) is transformed into a bidiagonal matrix, i.e.,
\begin{equation} \label{eq:eqChanBD}
  \mathbf{y}= \mathbf{Q}^H\mathbf{HPs} + \mathbf{Q}^H\Bar{\mathbf{z}} = \mathbf{B} \mathbf{s} + \mathbf{z}, 
\end{equation}
where $\mathbf{z}=\mathbf{Q}^H\Bar{\mathbf{z}}$ remains a zero-mean white Gaussian noise \cite{tse2005fundamentals}. After the post-processing step, the receiver can apply conventional equalizers for symbol detection according to $\mathbf{B}$.

\begin{figure}[tb]
\centering
\includegraphics[width=2.5in]{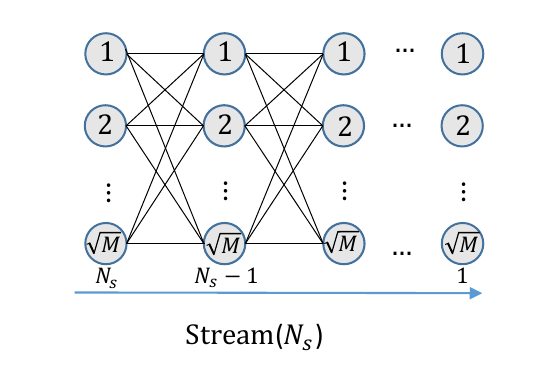}
\caption{The trellis model based on the real bidiagonal matrix: The in-phase and quadrature parts share the same trellis diagram.}
\label{FigT}
\end{figure}

An advantage property of $\mathbf{B}$ is that it is real-valued. Hence, the in-phase and the quadrature components of the received data i.e., $\mathbf{y}^{I} = \Re\{\mathbf{y}\}, \mathbf{y}^{Q} = \Im\{\mathbf{y}\}$, are decoupled. Specifically, the system model (\ref{eq:eqChanBD}) can be re-expressed as
\begin{equation}
    \begin{aligned}
\mathbf{y}^{I} & = \mathbf{Bs}^{I}+\mathbf{z}^{I}, \quad
\mathbf{y}^{Q} & = \mathbf{Bs}^{Q}+\mathbf{z}^{Q}, \\
\end{aligned}  
\end{equation}
where $\mathbf{s}^{I}$ and $\mathbf{s}^{Q}$ can be regarded as the $\sqrt{M}$-pulse amplitude modulation (PAM) signals{\footnote{Here, we consider that the underlying QAM constellation $\cal M$ is of a squared shape, e.g., $16$-QAM, such that the resultant PAM constellations have the same size $\sqrt M$.}}. Consequently, the equalization after post-processing can be represented by two identical fully-connected $\sqrt{M}$-trellis diagrams as shown in Fig. \ref{FigT}, where the $ \sqrt{M}$ (instead of $M$) states correspond to the $\sqrt{M}$-PAM symbols contained in either $\mathbf{y}^{I}$ or $\mathbf{y}^{Q}$, and the $N_s$ layers (depths) correspond to the number of information symbols. Denote ${\cal S}$ as the set of $\sqrt{M}$-PAM, i.e.,
\begin{equation}
  {\cal S} = \sqrt{\frac{3}{2(M-1)}}\left\{-\sqrt{M}+1,-\sqrt{M}+3,\cdots,\sqrt{M}-1\right\}.
\end{equation}
Note that $\mathbf{B}$ is not Toeplitz. Hence, the branch metrics associated with the state-transitions are different from layer to layer, which is in contrast to a conventional trellis diagram, e.g., the trellis of a typical convolutional code \cite{viterbi1971convolutional}. 
Based on the trellis representation in Fig. \ref{FigT}, the Viterbi or Bahl-Cocke-Jelinek-Raviv (BCJR) algorithm \cite{BCJR} can be employed for ML/MAP sequence estimation, even though $\mathbf{B}$ does not process a diagonal structure. The error performance of the considered CBD transceiver closely depends on the ECCN of $\mathbf{B}$, i.e., ${\rm cond(\mathbf{B})} =\frac{\max B_{i,i}}{\min B_{i,i}}$.  We will return to this aspect in Section \ref{SEC:III}. 

We highlight that the considered CBD transceiver exhibits a favorable performance-complexity tradeoff that is important for practical applications. Owing to the real-valued nature of $\mathbf{B}$, the CBD transceiver exhibits scalable complexity with respect to the number of antennas. Specifically, the required complexity for equalization after the post-processing is only ${\cal O}\left(2 M N_s\right)$, which scales linearly with the number of information symbols $N_s$ and constellation size $M$. In contrast, other decompositions generally require equalization complexity that grows exponentially with $N_s$ rendering them unsuitable for large-scale systems. In fact, the multiple-diagonal decomposition with width $D$ is fundamentally analogous to the SD \cite{SD}, whose primary goal is to reduce the search space dimensionality. The equalization for such decomposed channel matrices generally requires complexity around ${\cal O}\left(M^{D-1} N_s\right)$. Compared to the CBD transceivers, such decomposition requires much higher complexity that grows exponentially with $D$. 
More importantly, as we will demonstrate later, CBD transceivers can achieve near-optimal achievable rate performance after proper optimization. Therefore, CBD transceivers exhibit a good performance-complexity tradeoff that is important for the application of next generation wireless networks.



\section{Achievable Rate Analysis of CBD-based MIMO-BICM Systems} \label{Performance and Comp}
Before delving into the details of the achievable rate analysis of CBD-based schemes, we first provide the achievable rate analysis for general MIMO-BICM systems.


\subsection{Achievable Rate of MIMO-BICM} \label{sec:MIMO-BICM}
For a non-iterative MIMO-BICM system with an ideal interleaver, the interleaved bits can be assumed to be independent \cite{BICM,BICM-ID}, each with equal probability 
\begin{equation}\label{eq:eqprobe}
    p(c_k=1)=p(c_k=0)=\frac{1}{2}.
\end{equation}
In a practical MIMO-BICM system, the receiver calculates the soft information for each interleaved bit independently, before passing it to the deinterleaver and decoder as shown in Fig.~\ref{BICM}. The achievable rate of a MIMO-BICM system is obtained by averaging the MI for each bit level over the channel and summing the results as follows \cite{biglieri2000bit,mckay2005capacity}
\begin{equation} 
\begin{aligned} \label{eq:CAPBICM}
  \rm R_{MIMO-BICM}
   & = \sum^{K}_{k=1}I(c_k;\Bar{\mathbf{y}}|\mathbf{H}),
\end{aligned}
\end{equation}
where $I(c_k;\Bar{\mathbf{y}}|\mathbf{H})$ is the MI of the bit $c_k,k\in\{1,2,\cdots,K\}$ and the received data $\Bar{\mathbf{y}}$, given the channel $\mathbf{H}$. The MI $I(c_k;\Bar{\mathbf{y}}|\mathbf{H})$ is given by
\begin{equation} 
\begin{aligned} \label{eq:BICM}
  I(c_k;\Bar{\mathbf{y}}|\mathbf{H})
   & =  1 - \mathbb E_{\mathbf{s},\Bar{\mathbf{y}},\mathbf{H}} \left\{ \log_2 \frac{\sum_{\mathbf{s}\in \mathcal{X}}p(\Bar{\mathbf{y}}|\mathbf{s},\mathbf{H})}{\sum_{\mathbf{s}\in \mathcal{X}^{k}}p(\Bar{\mathbf{y}}|\mathbf{s},\mathbf{H})}  \right\}.\\
\end{aligned}
\end{equation}
From~\eqref{eq:CAPBICM} and (\ref{eq:BICM}), it becomes obvious that the direct assessment of the achievable rate for MIMO-BICM systems can be difficult because of the large number of combinations, especially in the case of massive MIMO.

To develop an efficient method for computing the achievable rate of MIMO-BICM systems, we assume that the coded bits remain independent after interleaving and de-interleaving \cite{BICM}. In the case that the channel matrix is known, the \textit{a posterior probability} (APP) can be computed as
\begin{equation}\label{eq:pfromLLR}
  p(c_k|\Bar{\mathbf{y}})=\frac{1}{1+e^{(1-2c_k)L_k}}.
\end{equation} 
It follows that
\begin{equation}
  \mathcal{H}(c_k|\Bar{\mathbf{y}}) = \mathbb E_{c_k, \Bar{\mathbf{y}}}\{\log_2(-p(c_k|\Bar{\mathbf{y}}))\}.
\end{equation}
In the above equation, $L_k$ is the LLR defined as
\begin{equation}\label{LLR}
    L_k = \log \frac{p(c_k=1|\Bar{\mathbf{y}},\mathbf{H})}{p(c_k=0|\Bar{\mathbf{y}},\mathbf{H})}.
\end{equation}
This allows us to re-express (\ref{eq:BICM}) as
\begin{equation}\label{eq:RBICM1}
    I(c_k;\Bar{\mathbf{y}}|\mathbf{H}) = 1-\mathcal{H}(c_k|\Bar{\mathbf{y}},\mathbf{H}),
\end{equation}
where $\mathcal{H}(c_k|\Bar{\mathbf{y}},\bf H)$ denotes the conditional entropy. Thus, the achievable rate is 
\begin{equation}\label{eq:RBICM}
\begin{aligned}
    \rm R_{ MIMO-BICM} &= \sum^{K}_{k=1}\left( 1-\mathbb E_{c_k, \mathbf{z}, \mathbf{H}} \left\{\log_2 \left(1+e^{(1-2c_k)L_k} \right) \right\} \right).
\end{aligned}
\end{equation} 
Here, for the equivalent channel $\bf B$, the calculation of $L_k, k=1,2,\cdots,K$ is based on the entries in $\mathbf{B}$ \cite{CBDbib}. Note that (\ref{eq:RBICM}) provides a general performance metric adopting the LLRs for the MIMO soft demodulators/equalizers, which is independent of the channel coding in use.


\subsection{Non-uniqueness of CBD}

An intriguing property of the CBD is that it is generally not unique for a given matrix. We demonstrate this important fact via the following lemma.

\begin{lemma}\label{theorem:cbd}
    For any given matrix $\mathbf{H}$ that can be bidiagonalized by $\mathbf{H}=\mathbf{QBP}^H$, there exists a unitary matrix $\mathbf{\Omega} \in \mathbb C^{N_t\times N_t}$ ($\mathbf{\Omega}\mathbf{\Omega}^H = \mathbf{I}_{N_t}$ and $\mathbf{\Omega}\neq\mathbf{I}_{N_t}$) making $\mathbf{H}\mathbf{\Omega}^H=\Tilde{\mathbf{Q}}\Tilde{\mathbf{B}}\Tilde{\mathbf{P}}^H$ and $\mathbf{H} = \tilde{\mathbf{H}} \mathbf{\Omega} = \tilde{\mathbf{Q}}\tilde{\mathbf{B}} (\mathbf{\Omega}^H \tilde{\mathbf{P}})^H$ with $\Tilde{\mathbf{B}} \neq \mathbf{B}$.
\end{lemma}
\begin{proof}
For a channel matrix $\mathbf{H}$, right-multiplying by such a non-identity unitary matrix $\mathbf{\Omega}^H$ obtains $\tilde{\mathbf{H}}$, which generally alters the entries and the norm of the first column in $\mathbf{H}$, i.e., $\mathbf{h}_1 \neq \Tilde{\mathbf{h}}_1$ and $\|\mathbf{h}_1\|^2 \neq \|\Tilde{\mathbf{h}}_1\|^2$. Therefore, the first step of Householder-based bidiagonalization of $\mathbf{H}$ and $\Tilde{\mathbf{H}}$ generates the distinct entries $B_{1,1}$ and $\Tilde{B}_{1,1}$ ($B_{1,1}\neq \Tilde{B}_{1,1}$) in $\mathbf{B}$ and $\Tilde{\mathbf{B}}$, respectively. This concludes the proof.
\end{proof}


Since the bidiagonalization of a given channel matrix is not unique as outlined in Lemma \ref{theorem:cbd}, the ECCN and performance of the CBD transceivers may vary significantly according to the underlying structure of the decomposed matrices. Therefore, we are motivated to investigate the optimal CBD in the following subsection.



\subsection{Asymptotic Analysis of Achievable Rate}\label{SEC:III}
In this subsection, we seek to optimize the bidiagonal structure for an improved asymptotic achievable rate of MIMO-BICM systems. Adopting the \textit{log-max} approximation \cite{SD}, we rewrite (\ref{LLR}) as
\begin{equation}\label{eq:LLR}
\begin{aligned}
    L_k &= \log \frac{\sum_{\mathbf{s}\in\mathcal{X}^{1}}p(\bf{y}|{\bf s ,\bf H})}{\sum_{\mathbf{s}\in\mathcal{X}^{0}}p(\bf{y}|{\bf s ,\bf H})} \approx \log \frac{\mathop{\max}_{\mathbf{s}\in \mathcal{X}^{1}} p(\bf{y}|{\bf s ,\bf H})}{\mathop{\max}_{\mathbf{s}\in \mathcal{X}^{0}} p(\bf{y}|{\bf s ,\bf H})} \\
& = \frac{1}{\sigma^2_z}\left\{ \mathop{\min}_{\mathbf{s}\in \mathcal{X}^{0}}\|\mathbf{y-Hs}\|^2- \mathop{\min}_{\mathbf{s}\in \mathcal{X}^{1}}\|\mathbf{y-Hs}\|^2\right\}.
\end{aligned}
\end{equation}
Based on the above, we are able to derive a lower-bound of the instantaneous achievable rate given the channel $\bf H$ as shown below.

\begin{theorem}
    Given the channel $\mathbf{H}$, the achievable rate of the MIMO-BICM system in both low and high SNR regimes is lower-bounded by
    \begin{equation}\label{eq:the1}
    \begin{aligned}
        \rm R_{L} \triangleq &\sum_{i=1}^{N_s}\sum_{k=1}^{Q_m}\left(1-\log_2\left(1+e^{-d^2_{\rm min}\frac{\mathbf{h}_i^H\mathbf{h}_i}{\sigma^2_z}}\right)\right),
    \end{aligned}
    \end{equation}
where $d^2_{\rm min} = \frac{6}{M-1}$ is the MED of the constellation set $\mathcal{M}$ and $\mathbf{h}_i$ is the $i$-th column of $\mathbf{H}$. 
\end{theorem}
\begin{proof}
    See Appendix \ref{App:A}.
\end{proof}

Next, we propose to improve the achievable rate by maximizing~\eqref{eq:the1}. 
Adopting the precoder $\mathbf{P}$ and the equalizer $\mathbf{Q}$ in (\ref{eq:CBD}), we can re-express (\ref{eq:the1}) as
\begin{equation}\label{eq:HP}
    {\rm R_{L}} = \sum_{i=1}^{N_s}\sum_{k=1}^{Q_m}\left(1-\log_2\left(1+e^{-d^2_{\rm min}\frac{\mathbf{b}_i^H\mathbf{b}_i}{\sigma^2_z}}\right)\right),
\end{equation}
where $\mathbf {b}_i$ being the $i$-th column of $\bf B$. Based on~\eqref{eq:HP}, we have the two following propositions.
\begin{proposition}\label{proposition:1}
    In low SNR regimes, the maximization of the achievable rate is equivalent to solving 
    \begin{equation}
    \begin{array}{cl}
    \underset{\mathbf{B}}{\operatorname{max}} & \sum^{N_s}_{i=1}\mathbf{b}_i^H\mathbf{b}_i \\ 
    \text {\rm s.t. }  & (\ref{eq:horn1}),{\rm Tr}(\mathbf{P}^H\mathbf{P})\leq N_t,
    \end{array} 
    \end{equation}  
    and the optimal CBD of $\bf H$ is given by the SVD solution with $\bf B$ being the diagonal matrix containing the singular values. In such a case, the channel matrix $\mathbf{H}$ with larger ECCN are beneficial for achieving an improved achievable rate under a fixed channel gain ${\rm Tr}\left({\bf H}^H {\bf H}\right)$.
\end{proposition}
\begin{proof}
Considering $\sigma_z^2\rightarrow \infty$ and applying the Taylor series expansion to~\eqref{eq:HP}, we obtain the following optimization problem for the CBD optimization in the low SNR regime.
\begin{equation}
\begin{array}{cl}
\underset{\mathbf{B}}{\operatorname{max}} & \sum^{N_s}_{i=1}\sum^{Q_m}_{k=1}\left(\frac{d_{\rm min}^2\mathbf{b}_i^H\mathbf{b}_i}{2\ln2 \sigma^2_z} - o(\frac{1}{\sigma^2_z})\right) \\ 
\text {\rm s.t. }  & (\ref{eq:horn1}),{\rm Tr}(\mathbf{P}^H\mathbf{P})\leq N_t,
\end{array} 
\end{equation}   
where $o(\frac{1}{\sigma^2_z})$ denotes the least significant terms on the order of $\frac{1}{\sigma^2_z}$. Note that $d^2_{\rm min} = \frac{6}{M-1}$ is a fixed constant for a given QAM constellation set $\cal M$. Thus, the above optimization is equivalent to maximizing ${\rm Tr}(\mathbf{HP}^H\mathbf{HP})$ with the power constraint. Recall that $\mathbf{H = QBP}^H$ and therefore we have $(\mathbf{HP})^H\mathbf{HP}=\mathbf{B}^H\mathbf{B}$. Consequently, the above problem reduces to maximizing ${\rm Tr}(\mathbf{B}^H\mathbf{B})=\sum^{N_s}_{i=1}\mathbf{b}_i^H\mathbf{b}_i$, where the optimal solution is given by the SVD, i.e., the off-diagonal entries in the bidiagonal matrix $\mathbf{B}$ are zeros. In this case, the achievable rate of the system will be dominated by several large eigenmodes that exhibit relatively high effective SNR, which suggests that a large ECCN is generally beneficial under a fixed channel gain ${\rm Tr}\left({\bf H}^H {\bf H}\right)$.

\end{proof}

\begin{proposition}\label{proposition:2}
    In high SNR regimes, the maximization of (\ref{eq:HP}) is equivalent to solving  
    \begin{equation} \label{eq:temp1}
\begin{array}{cl}
\underset{\mathbf{B}}{\operatorname{max}}\underset{1\leq i \leq N_s}{\min} & \|B_{i,i}\|^2\\
\text {\rm s.t.} & {\rm Tr}(\mathbf{P}^H\mathbf{P}) \leq N_t,\\ 
&(\mathbf{HP})^H\mathbf{HP}=\mathbf{B}^H\mathbf{B},
\end{array} 
\end{equation}
where the optimal solution is achieved when the diagonal entries of $\bf B$ are of the same value, resulting in a smaller ECCN of the given channel $\mathbf{H}$.  
\end{proposition}
\begin{proof}
     When $\sigma_z^2\rightarrow 0$, the function $\sum_i \log_2(1+\exp(-\frac{d^2_{\rm min}\mathbf{b}_i^H\mathbf{b}_i}{\sigma^2_z}))$ is potentially dominated by the smallest term $\mathbf{b}_i^H\mathbf{b}_i$. From \cite[Sec. III-D]{EqualQR} it has been proved that the minimum distance of a signal after passing through an equivalent channel can be bounded by the diagonal entries of the corresponding equivalent upper triangular matrix. Notably, the matrix $\mathbf{B}$ represents a specific instance of such an upper triangular matrix. Thus, the problem can be approximately by following:
\begin{equation} \label{eq:te1}
\begin{array}{cl}
\underset{\mathbf{B}}{\operatorname{max}}\min & d^2_{\rm min}\mathbf{b}^H_i\mathbf{b}_i \\ 
\text {\rm s.t. }  & {\rm Tr}(\mathbf{P}^H\mathbf{P})\leq N_t,\\
&(\mathbf{HP})^H\mathbf{HP}=\mathbf{B}^H\mathbf{B}.
\end{array} 
\end{equation}  
It follows that
\begin{equation}
    \mathbf{b}_i^H\mathbf{b}_i = \sum_{j=i}^{i-1}\|B_{j,i}\|^2\geq\|B_{i,i}\|^2.
\end{equation} 
Maximizing this lower bound is equivalent to solving ~\eqref{eq:temp1}. Thus, the optimal solution to ~\eqref{eq:temp1} is obtained when all diagonal entries of $\mathbf{B}$ are equal. 


\end{proof}

The physical interpretations of Propositions \ref{proposition:1} and \ref{proposition:2} are highly insightful as follows: In the low SNR regime (or equivalently, at low channel coding rates), the transmitted energy should be concentrated on the dominant channel to maximize the achievable rate. In this case, the adverse effects of channel ill-conditioning are mitigated. Conversely, in the high SNR regime (or equivalently, at high channel coding rates), the achievable rates would be benefited by ensuring all sub-channels having reasonably high effective SNR.
More importantly, the above propositions hold for general MIMO-BICM systems. However, they confirm that no universal channel decomposition solution e.g., $\bf B$, is optimal across the entire SNR regime.

The above conclusions, while seemingly intuitive, remain empirically unverified within MIMO-BICM systems. Moreover, we observe that $\bf B$ plays an important role in the achievable rate performance. In fact, finding an optimal solution for optimizing CBD structure is still an open question, where preliminary results have demonstrated that bidiagonalization with equal diagonal entries is generally unattainable \cite[Sec. V]{jiang2005geometric}. In what follows, we conduct a quantitative analysis of the CBD, aiming to provide capacity approaching designs based on the insights obtained above. 

\section{The GP-CBD Scheme} \label{SEC:GP-CBD}
In this section, we investigate the optimal bidiagonal elements and propose a novel GP-CBD scheme to approach the optimal bidiagonal decomposition. Specifically, we construct the target $\mathbf{B}$ leveraging Givens rotation \cite{golub2013matrix}, which allows for unitary transform between any two eigen-subchannels to alter the ECCN, and is therefore more flexible than the Householder transformation. 

\subsection{CBD Design for $2\times2$ MIMO-BICM}\label{SUBSEC:GP-CBD}
We first investigate a $2\times 2$ channel with effective eigen-subchannel coefficients $\lambda_1$ and $\lambda_2$ to obtain valuable insights. It has been proved that the error-free DFE (error-free interference cancellation based on genie-aided assumption) can achieve the ML performance \cite{EqualQR,1268365}. Therefore, we leverage the DFE concept to help analyzing the diagonal entries by assuming that the interference term in the equivalent bidiagonal matrix can be perfectly eliminated.
In this case, the LLR [cf. (\ref{eq:LLR})] can be approximated as 
\begin{equation}
    L_k = \rho_{o,i} \Delta_k,
\end{equation}
where $\Delta_k = \left\{ \mathop{\min}_{{s}\in \mathcal{M}^{0}} \|\Tilde{y}_i-{s}\|^2 - \mathop{\min}_{{s}\in \mathcal{M}^{1}} \|\Tilde{y}_i-{s}\|^2 \right\}$ is determined by the specific noise, $\rho_{o,i} \triangleq \frac{{B}^2_{i,i}}{\sigma_z^2}$ is the post-processing SNR of the $i$-th layer \cite[Sec. IV]{palomardaniel2006MIMO}, and $\Tilde{y}_i = \frac{y_i}{{B}_{i,i}}$. In high SNR regimes, $\Delta_k$ can be calculated by Monte Carlo simulations, and is proportional to the MED $d^2_{\rm min}$  \cite{olmos2012use}. Hence, in a $2\times 2$ system, we maximize (\ref{eq:RBICM}) by approximating $\Delta_k$ by $d^2_{\rm min}$, which leads to
\begin{equation}\label{eq:RBICMnoise}
\begin{array}{cl}
\underset{{B}_{i,i}}{\operatorname{max}} & \sum_{k=1}^{Q_m}\sum_{i=1}^{2}\left( 1-\mathbb E_{c_k} \left\{\log_2 (1+e^{-d^2_{\rm min}\frac{{B}_{i,i}}{\sigma^2_z}} ) \right\} \right) \\  
\text {\rm s.t. }  & \lambda_{2}\leq B_{2,2} \leq B_{1,1} \leq \lambda_{1}, B_{1,1}B_{2,2}=\lambda_{1}\lambda_{2},
\end{array}  
\end{equation} 
where the constraint is from (\ref{eq:horn1}). Here, the term $\mathbb E_{c_k}\{\cdot\}$ is dependent of $c_k$, then maximizing (\ref{eq:RBICMnoise}) is equivalent to 
\begin{equation}\label{eq:2t2}
\begin{array}{cl}
\underset{{B}_{i,i}}{\operatorname{min}} & \sum^{2}_{i=1}\left(\log_2(1+e^{-\frac{d^2_{\rm min}{B}_{i,i}^2}{\sigma^2_z}})\right) \\ 
\text {\rm s.t. }  & \lambda_{2}\leq B_{2,2} \leq B_{1,1} \leq \lambda_{1}, B_{1,1}B_{2,2}=\lambda_{1}\lambda_{2},
\end{array} 
\end{equation}
(\ref{eq:2t2}) quantifies the gap between the achievable rate and saturation rate $Q_mN_s$ (due to finite-alphabet inputs) at a specific SNR. The solution can be obtained by the following theorem.

\begin{theorem}\label{theorem:2}
For a $2\times 2$ MIMO-BICM channel with two eigen-subchannels $\lambda_{i},i=1,2$, we define $\mu = (\frac{d^2_{\rm min}\lambda_{1}^2}{\sigma_z^2})(\frac{d^2_{\rm min}\lambda_{2}^2}{\sigma_z^2})$ and the threshold $\nu = 1.7$. If $\mu\geq \nu$, the diagonal entries in $\mathbf{B}$ are $B_{1,1}= B_{2,2}= \sqrt{\lambda_1\lambda_2}$; otherwise, the solution can be obtained through a 1-D search within the feasible region.
\end{theorem}
\begin{proof}
To simplify the form of problem (\ref{eq:2t2}), we define $x =B_{1,1}^2 \frac{d^2_{\rm min}}{\sigma_z^2}$. Then, (\ref{eq:2t2}) can be re-expressed by a single-variable optimization problem over $x$:
\begin{equation} \label{eq:ProblemEffe}
\begin{array}{cl}
\underset{x}{\operatorname{min}} & g(x) = \log_2(1+e^{-x})+\log_2(1+e^{-\frac{\mu}{x}}) \\ 
\text { s.t. }  & \lambda_1^2\frac{d^2_{\rm min}}{\sigma_z^2} \geq x \geq \sqrt{\mu}. 
\end{array}    
\end{equation}
with $\mu = \prod_{i=1}^2 (\frac{d^2_{\rm min}\lambda_{i}^2}{\sigma_z^2})$. Here, $g(x)$ represents the gap between the achievable rate and the saturation rate. 
This constraint leverages the symmetry by reducing (\ref{eq:2t2}) to a function of $x$, allowing us to solve $g(x)$ at the feasible region. The first-order derivatives of $g(x)$ is
\begin{equation}\label{eq:firstderive}
\begin{aligned}
        g'(x) = \frac{1}{\ln 2}\frac{-(1+e^{\frac{\mu}{x}})+\frac{\mu}{x^2}(1+e^{x})}{(1+e^{\frac{\mu}{x}})(1+e^{x})}
\end{aligned}
\end{equation} and the second-order derivation of $g(x)$ is
\begin{equation}
\begin{aligned}
    & g''(x)= \frac{1}{\ln 2}\left(\frac{e^{x}}{(1+e^{x})^2}+\frac{-\frac{2\mu}{x^3}(1+e^{\frac{\mu}{x}})+\frac{\mu^2}{x^4}e^{\frac{\mu}{x}}}{(1+e^{\frac{\mu}{x}})^2}\right).
\end{aligned}
\end{equation}
It follows directly from (\ref{eq:firstderive}) that $x = \sqrt{\mu}$ is a stationary point, as $g'(\sqrt{\mu})=0$. The second derivative $g''(x)$ at $x=\sqrt{\mu}$ is 
\begin{equation}
\begin{aligned} \label{eq:2time2numeator}
    &g''(\sqrt{\mu})= \frac{2e^{\sqrt{\mu}}-\frac{2}{\sqrt{\mu}}(1+e^{\sqrt{\mu}})}{\ln 2 (1+e^{\sqrt{\mu}})^2}. 
\end{aligned}
\end{equation}
If $g''(\sqrt{\mu}) \geq 0$ holds, then $g'(x)\geq 0$, confirming that $\sqrt{\mu}$ is a global minimum point. We now turn our attention to the numerator in (\ref{eq:2time2numeator}) to investigate its positivity. It is straightforward to show that the function $g''(\sqrt{\mu})$ is strictly increasing and thus has only one zero point
\begin{equation}
    e^{\sqrt{\mu}}-\frac{1}{\sqrt{\mu}}(1+e^{\sqrt{\mu}}) = 0,
\end{equation}
which occurs at $\mu = 1.7$ according to numerical calculation. We set the threshold $\nu = 1.7$. If $\mu \geq \nu$, we have the optimal solution $x^2 = \mu$, i.e., $B_{1,1} = B_{2,2}= \sqrt{\lambda_1\lambda_2}$; if $\mu < \nu$, the diagonal entries can be obtained through the 1-D search within the feasible region.
\end{proof}

Notice that the above solution only calculates the diagonal entries of $\bf B$. After determining those entries, we employ the Givens rotation to obtain the bidiagonal matrix as 
\begin{equation}\label{eq:wb}
\mathbf{B} = \hat{\mathbf{G}}\left[\begin{matrix}
    \lambda_1 & 0 \\
    0 & \lambda_2
  \end{matrix} \right]\tilde{\mathbf{G}}=\left[\begin{matrix}
    B_{1,1} & m \\
    0 & B_{2,2}
  \end{matrix} \right],
\end{equation}
where the Givens matrices are
\begin{equation} \label{eq:t2}
  \hat{\mathbf{G}} = \frac{1}{B_{1,1}} \left[\begin{matrix}
    c\lambda_1 & s\lambda_2 \\
    -s\lambda_2 & c\lambda_1
  \end{matrix} \right], \quad \tilde{\mathbf{G}} = \left[\begin{matrix}
    c & -s \\
    s & c
  \end{matrix}\right].
\end{equation}
with $c = \sqrt{\frac{|B_{1,1}|^2-|\lambda_2|^2}{|\lambda_1|^2-|\lambda_2|^2}}$, $s=\sqrt{1-c^2}$, and $m=\frac{sc(|\lambda_1|^2-|\lambda_2|^2)}{B_{1,1}}$. Here, $B_{2,2}$ and $B_{1,1}$ are the parameters obtained from Theorem \ref{theorem:2}.

\textit{Remark 1: } In fact, Theorem \ref{theorem:2} goes beyond the conclusions of Proposition \ref{proposition:2} by proposing a quantitative criterion of the problem. Specifically, we employ $g(x)$ to quantify the gap between the achievable rate and the saturation rate. Given the modulation order $M$ and noise power $\sigma_{z}^2$, the condition of a ``well-conditioned'' channel is established as $\mu \geq 1.7$ in a $2\times 2$ MIMO-BICM system. For such cases, performing the appropriate rotation on the eigen-subchannels and setting equal diagonal entries in $\mathbf{B}$ minimizes $g(x)$, thereby reducing the gap between the achievable rate and the saturation rate. To better illustrate the effect of $\mu$, we depict the relationship of $\mu$ and $g(x)$ in Fig. \ref{Fig:gx}. It is observed that for a small value $\mu<1,7$ (e.g., $\mu=1$), optimizing $g(x)$ results in minor gap reduction, and yields negligible rate gains ($<0.05$). In contrast, in the well-conditioned channel with large $\mu$ (e.g., $\mu=8$), the gain becomes significant by reducing the ECCN through (\ref{eq:wb}).

\textit{Remark 2: }
The threshold $\nu=1.7$ offers an alternative perspective on the performance degradation of the SVD-based scheme with a large ECCN at high SNRs. In Fig. \ref{Fig:gx}, the horizontal dashed line intersects $g(x)$ at two points $\frac{d^2_{\rm min}\lambda_{1}^2}{\sigma_z^2}$ and $\frac{d^2_{\rm min}\lambda_{2}^2}{\sigma_z^2}$. Given the channel $\mathbf{H}$ and $d_{\rm min}^2$, then the ECCN ($\frac{\lambda_1}{\lambda_2}$) of the SVD is fixed. As the SNR increases, the optimization space of $g(x)$ (the gap of the dash line and the minimal point $-\star-$) grows larger, which drives the SVD-based scheme far from the optimal achievable rate performance.

\begin{figure}[t!]
\centering
\includegraphics[width=3in]{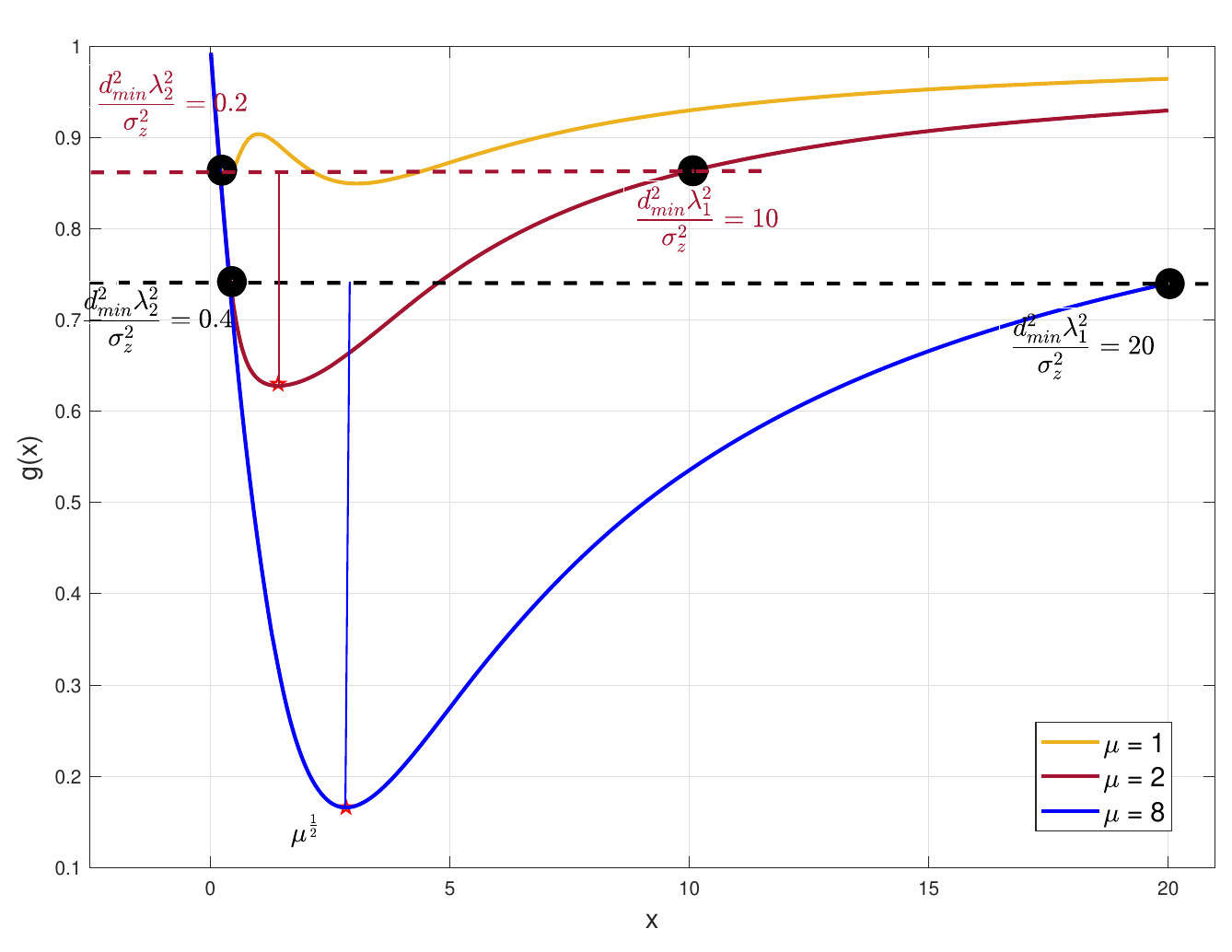}
\caption{ The function of $g(x)$ with different values of $\mu$. The two black dash lines show ECCN $=7.07$ of the SVD at different SNRs.}
\label{Fig:gx}
\end{figure}

\subsection{GP-CBD Design}
We next extend the previous analysis to the  general MIMO-BICM systems with $N_r > 2$ and $N_t > 2$.
In this case, (\ref{eq:2t2}) extends to 
\begin{equation}\label{eq:NtN}
\begin{array}{cl}
\underset{{B}_{i,i}}{\operatorname{min}} & \sum^{N_s}_{i=1}\left(\log_2(1+e^{-\frac{d^2_{\rm min}{B}_{i,i}^2}{\sigma^2_z}})\right) \\ 
\text {\rm s.t. }  & (\ref{eq:horn1}) ,
\end{array} 
\end{equation}
which is challenging to solve. To make (\ref{eq:NtN}) tractable, we divide it into multiple subproblems [cf. (\ref{eq:2t2})] by strategically pairing the eigen-subchannels [cf. (\ref{eq:svd})] two by two, i.e., $\{\lambda_i,\lambda_j\}$. Each eigen-subchannel is paired only once. Without loss of generality, we consider $\lambda_i>\lambda_j$. The pairing follows $(i,j), i,j \in \{1,2,\cdots,N_s\}$, forming what we refer to as ``paired subchannel''. We then apply Givens rotation to each paired subchannel to obtain a set of subbidiagonal matrices, which collectively form the final $\mathbf{B}$. This approach offers two key benefits: i) strategic eigen-subchannel pairing facilitates precise tuning of the diagonal entries in each sub-bidiagonal matrix to minimize (\ref{eq:NtN}); ii) the resulting bidiagonal matrix exhibits a block-diagonal structure, enabling parallel processing at the receiver and significantly reducing computational complexity, referred to as the GP-CBD.   

The pairing of eigen-subchannels is important. The eigen-subchannels are ordered as $\lambda_1\geq\lambda_2\geq\cdots\geq\lambda_{N_s}$. As evidenced by \cite{EqualQR}, optimizing the worst subchannel necessitates a high SNR; otherwise, insufficient SNR conditions lead to performance degradation. In contrast, (\ref{eq:optimization}) targets specific eigen-subchannels that can be optimized to enhance the achievable rate. From Remark 1, prioritizing rotation on well-conditioned paired subchannels potentially enhances performance. Therefore, our focus is on identifying well-conditioned subchannels that is demarcated by $\lambda_{N},1<N\leq N_s$ given a certain SNR. According to Theorem~\ref{theorem:2}, $\lambda_{N}$ can be determined by solving the following problem:
\begin{subequations}\label{eq:optimization}
\begin{align}
    & \underset{N}{\text{max}} & & \frac{\lambda_1^2}{\lambda_N^2} \label{eq:objective} \\
    & \text{\rm s.t. } & & (\frac{d^2_{\rm min}\lambda_1^2}{\sigma_z^2})\cdot(\frac{d^2_{\rm min}\lambda_N^2}{\sigma_z^2})\geq 1.7. \label{eq:constraint2}
\end{align}
\end{subequations}
According to Theorem \ref{theorem:2},  (\ref{eq:constraint2}) implies that two eigen-subchannels selected from $\lambda_{1},\cdots,\lambda_{N}$ generally form a well-conditioned pair. After determining $N$, we pair the eigen-subchannel as $(n, N - n + 1), n \in \{1, 2, \dots, \lfloor \frac{N}{2} \rfloor\}$. These paired subchannels are denoted as $\mathbf{\Tilde{\Lambda}}_1={\rm diag}\{\lambda_1, \lambda_{N}\},\mathbf{\Tilde{\Lambda}}_2={\rm diag}\{\lambda_2, \lambda_{N-1}\}, \cdots, \mathbf{\Tilde{\Lambda}}_{\lfloor \frac{N}{2} \rfloor}={\rm diag}\{\lambda_{\lfloor \frac{N}{2} \rfloor}, \lambda_{N-\lfloor \frac{N}{2} \rfloor+1}\}$. Performing Givens rotations on these eigen-subchannels generally improves the performance. The eigen-subchannels $\lambda_{N+1},\cdots,\lambda_{N_s}$ remain diagonal, as their rotation provides negligible gain. This case can be viewed as a special case of CBD, where the subdiagonal entries are zero. This approach maximizes the MED of the smallest optimizable eigen-subchannel $\lambda_N$, thereby enhancing performance. Therefore, the optimal solution to (\ref{eq:NtN}) is approximated through a feasible way by a strategic pairing approach and optimal rotation of each paired subchannel. We shall note that this pairing scheme, while not necessarily optimal, offers a straightforward pairing criterion, facilitating practical implementation. Moreover, this process adaptively adjusts the ECCN under varying SNRs, thereby enhancing the achievable rate and BER performance, as demonstrated by numerical results. 

We then rearrange these paired subchannels to facilitate the GP-CBD process. This rearrangement is accomplished by left- and right-multiplying the appropriate permutation matrices \cite[Ch-1.2.8]{golub2013matrix}. As an example, we consider the case where $N_s = 8$ with ${\lambda_1 \geq \cdots\geq\lambda_8}$. If the solution of (\ref{eq:optimization}) is $N=6$, we group these eigen-subchannels into four paired subchannels $\mathbf{\Tilde{\Lambda}}_1={\rm diag}\{\lambda_1, \lambda_6\}$, $\mathbf{\Tilde{\Lambda}}_2={\rm diag}\{\lambda_2, \lambda_5\}$, $\mathbf{\Tilde{\Lambda}}_3={\rm diag}\{\lambda_3, \lambda_4\}$, and
$\mathbf{\Tilde{\Lambda}}_4={\rm diag}\{\lambda_7, \lambda_8\}$. Then
\begin{equation}\label{eq:SVDRERANGE}
\begin{aligned}
   \Tilde{\mathbf{\Lambda}} &= \mathbf{\Gamma}_l \mathbf{\Lambda} \mathbf{\Gamma}_r = \rm blkdiag (\mathbf{\Tilde{\Lambda}}_1,\mathbf{\Tilde{\Lambda}}_2,\mathbf{\Tilde{\Lambda}}_3,\mathbf{\Tilde{\Lambda}}_4),
\end{aligned}
\end{equation}
where $\mathbf{\Gamma}_l$ and $\mathbf{\Gamma}_r$ are the desired permutation matrices. Then Givens matrices are performed on each $\mathbf{\Tilde{\Lambda}}_i$ to obtain subbidiagonal matrices. For the general case, the whole channel matrix $\mathbf{H}$ can be decomposed as
\begin{equation} \label{eq:PLCBD}
\begin{aligned}
  & \mathbf{H} = \mathbf{U\Lambda V}^H = \underbrace{\mathbf{U}\mathbf{\Gamma}^H_l}_{\Tilde{\mathbf{U}}}\Tilde{\mathbf{\Lambda}} \underbrace{\mathbf{\Gamma}^H_r \mathbf{V}^H}_{\Tilde{\mathbf{V}}^H} = \Tilde{\mathbf{U}}\Tilde{\mathbf{\Lambda}}\Tilde{\mathbf{V}}^H\\ 
  & = \underbrace{\Tilde{\mathbf{U}}
  \begin{bmatrix} 
\hat{\mathbf{G}}_1 & & \\
 & \ddots &  \\
 &  & \hat{\mathbf{G}}_T \\
\end{bmatrix}}_{\mathbf{Q}}
\underbrace{\begin{bmatrix} 
\mathbf{B}_1 &  & \\
 & \ddots &  \\
 &  & \mathbf{B}_T \\
\end{bmatrix}}_{\mathbf{B}_{\rm GP}} 
\underbrace{\begin{bmatrix} 
\Tilde{\mathbf{G}}_1 & & \\
  & \ddots &  \\
 &  & \Tilde{\mathbf{G}}_T \\
\end{bmatrix}\Tilde{\mathbf{V}}^H}_{\mathbf{P}^H} \\
    & = \mathbf{Q}\mathbf{B}_{\rm GP}\mathbf{P}^H,
\end{aligned}
\end{equation}
where $\mathbf{Q}$ and $\mathbf{P}$ are unitary matrices. Employing $\mathbf{P}$ as the precoder, the effective system model is $\Bar{\mathbf{y}} = \mathbf{Q}\mathbf{B}_{\rm GP}\mathbf{s} + \Bar{\mathbf{z}} $, which can be expressed as
\begin{equation}
\Bar{\mathbf{y}} = \begin{bmatrix}\mathbf{Q}_1 \cdots \mathbf{Q}_T\end{bmatrix}
\begin{bmatrix} 
            \mathbf{B}_1 &  & \\
                & \ddots &  \\
            &  & \mathbf{B}_T \\
    \end{bmatrix}\begin{bmatrix}\mathbf{s}_1 \\ \vdots \\ \mathbf{s}_T\end{bmatrix}+\begin{bmatrix}\Bar{\mathbf{z}}_1 \\ \vdots \\ \Bar{\mathbf{z}}_T\end{bmatrix},
\end{equation}
with $\mathbf{Q}_t \in \mathbb C^{N_r\times 2},\mathbf{B}_t \in \mathbb R^{2\times 2},t=1,2,\cdots,T$. Notice that $\mathbf{Q}^H_t\mathbf{Q}_t = \mathbf{I}_2$ and $\mathbf{Q}_t^H$ can be utilized as the sub-equalizer of each subchannel. Then, the effective signal transmission model becomes
\begin{equation}\label{eq:gpcbd}
    \mathbf{y}_t = \mathbf{B}_t \mathbf{s}_t + \mathbf{z}_t,\quad,t=1,2\cdots,T,
\end{equation}
where $\mathbf{y}_t = \mathbf{Q}_t^H\Bar{\mathbf{y}}, \mathbf{z}_t = \mathbf{Q}_t^H\Bar{\mathbf{z}},t=1,2,\cdots,T$. Hence, the received signal vector can be partitioned into $T$ blocks that can be processed independently (i.e., in parallel). 

On top of the GP-CBD given in~\eqref{eq:gpcbd}, we can consider power allocation among different eigen-subchannels. Let $\mathbf{\Phi} \in \mathbb R^{N_s\times N_s}$ be the diagonal power allocation matrix, which is calculated according to the subchannels of $\bf B_{\rm GP}$, i.e., $\rm diag\{\bf B_{\rm GP}\}$. Then we redefine the precoding matrix $\mathbf{F}$ by 
\begin{equation}
    \mathbf{F} \triangleq \mathbf{P\Phi}^{\frac{1}{2}}.
\end{equation}
Then we can re-express \ref{eq:gpcbd} as 
\begin{equation}
    \mathbf{y}_t = \mathbf{B}_t\mathbf{\Phi}_t^{\frac{1}{2}} \mathbf{s}_t + \mathbf{z}_t,\quad,t=1,2\cdots,T,
\end{equation}
where ${\bf \Phi}_t^\frac{1}{2}$ is the power allocation submatrix of each subbidiagonal matrix. Such a parallel bidiagonalization process would not only produce an improved achievable rate, but also facilitate the parallel equalization/detection of multiple data streams at the receiver. The whole algorithm framework is outlined in Algorithm 1. 

\begin{algorithm}[t!]
    \caption{The GP-CBD Scheme}
    \renewcommand{\algorithmicrequire}{\textbf{Input:}}
    \renewcommand{\algorithmicensure}{\textbf{Output:}}
    \begin{algorithmic}[1]
    \Require The noise variance $\sigma^2_z$; the MED $d^2_{\rm min}$; the eigen subchannels $\lambda_1, \lambda_2,\cdots,\lambda_{N_s}$; 
    \Ensure The GP-CBD decomposition.
    \State Set the threshold $\nu = 1.7$.
    \For{$m =1:N_s$}
    \State Set $\mu=\left( \frac{d^2_{\rm min}\lambda^2_1}{\sigma_z^2}\right)\cdot\left( \frac{d^2_{\rm min}\lambda^2_m}{\sigma_z^2}\right)$.
    \If{$\mu < \nu$}
    \State Break;
    \EndIf    
    \EndFor
    \State Set $N=m-1$.
    \State Obtain the paired subchannels as $(n, N - n + 1)$ for $n \in \{1, 2, \dots, \lfloor \frac{N}{2} \rfloor\}$, while leaving other eigen-subchannels unchanged.
    \State Perform permutation and Givens rotation on these paired subchannels as (\ref{eq:PLCBD}) to achieve the GP-CBD.
    \end{algorithmic}\label{alg.simplified2}
\end{algorithm}

\begin{table*}[htb]
\caption{COMPLEXITY COMPARISON OF DIFFERENT SCHEMES}
\label{table.env}
\centering
\begin{tabular}{|c||c|c|c|}
 \hline
   & \textbf{Design Complexity} & \textbf{Implementation Complexity} & \textbf{equivalent channel} \\
 \hline
   \textbf{GP-CBD} & $\mathcal{O}(N_tN_r^2)$ & $\mathcal{O}(2M)$ & \textbf{Real Bidiagonal}\\
 \hline
   \textbf{CBD in \cite{CBDbib}
} & $\mathcal{O}(N_tN_r^2)$ & $\mathcal{O}(MN_s)$ & \textbf{Real Bidiagonal} \\
 \hline
   \textbf{SVD-MMSE} & $\mathcal{O}(N_tN_r^2)$ & $\mathcal{O}(M)$ & \textbf{Real Diagonal}\\
 \hline
   \textbf{UCD-DP} & $\mathcal{O}(N_tN^2_r)$ & $\mathcal{O}(MN_s + N_s^2)$ & \textbf{Real Diagonal}\\
 \hline
\textbf{Scheme in \cite{maleki2024precoding}
} & $\mathcal{O}(N_tN^2_r)$ & $\mathcal{O}(M^2)$ & \textbf{Complex $2\times 2$ matrix} \\
 \hline
\textbf{Scheme in \cite{PL-CBD}} & $\mathcal{O}(N_tN^2_r)$ & $\mathcal{O}(JM), J=2,4$& \textbf{Real Bidiagonal}\\
 \hline
 
\end{tabular} \label{Table:ComplexAnalysis}
\end{table*}


\subsection{Complexity Analysis} \label{SUB:COMPLEX}

The proposed GP-CBD scheme includes both the transceiver design and implementation. Specifically, the transceiver design consists of the SVD operation and the submatrix transformations using Givens rotation, and these rotation can be performed in parallel. Then the design complexity is $\mathcal{O}(N_r^2N_t)$. For the implementation of the GP-CBD, both the precoder $\mathbf{F}$ and equalizer $\mathbf{Q}$ in (\ref{eq:CBD}) are inherently block diagonal and can be operated in parallel with a same computational complexity as the SVD-MMSE. The detection after equalization can be also performed in parallel, with a complexity $\mathcal{O}(2M)$. However, for other near-optimal schemes, such as, SD or the uniform channel decomposition (UCD) combining dirty paper coding (UCD-DP), they have a higher complexity as summarized in Table \ref{Table:ComplexAnalysis}. 



\section{Numerical Results}\label{SEC:results}

This section presents extensive numerical simulations to evaluate the ECCN, the achievable rate, and BER performance of the proposed GP-CBD scheme.
\begin{figure*}[!ht]
\begin{center}
    \begin{minipage}[b]{0.32\textwidth} 
        \centering
        \includegraphics[width=2.25in, height=1.7in]{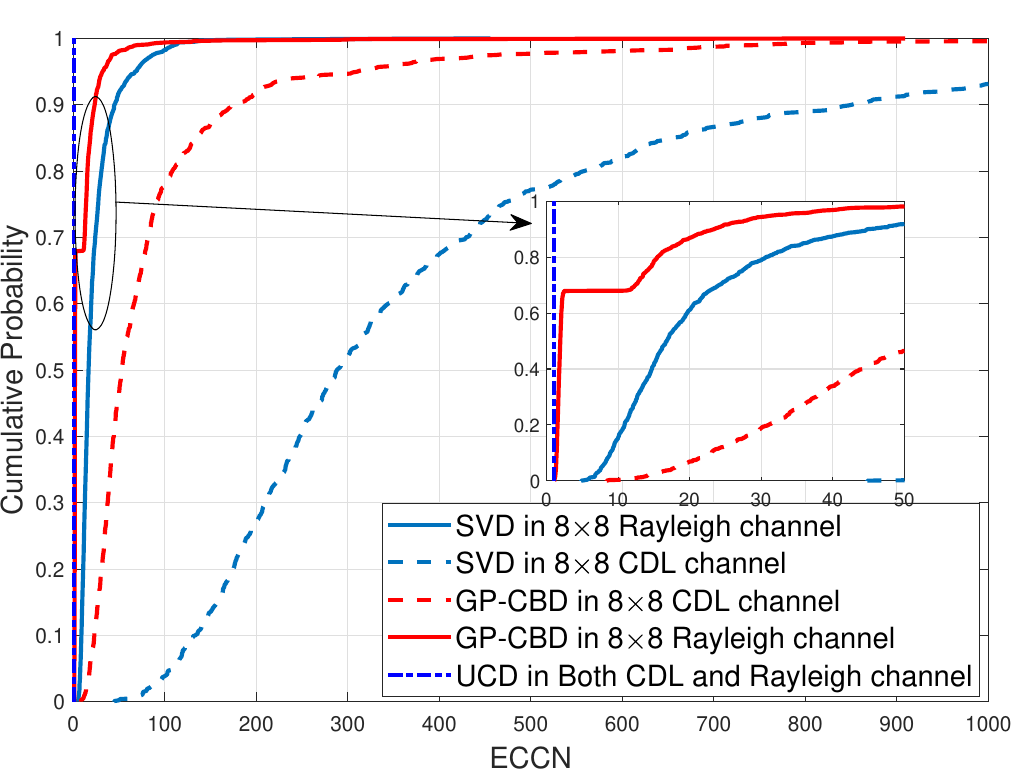}
        \subcaption{} 
    \end{minipage}
    \begin{minipage}[b]{0.32\textwidth} 
        \centering
        \includegraphics[width=2.25in, height=1.7in]{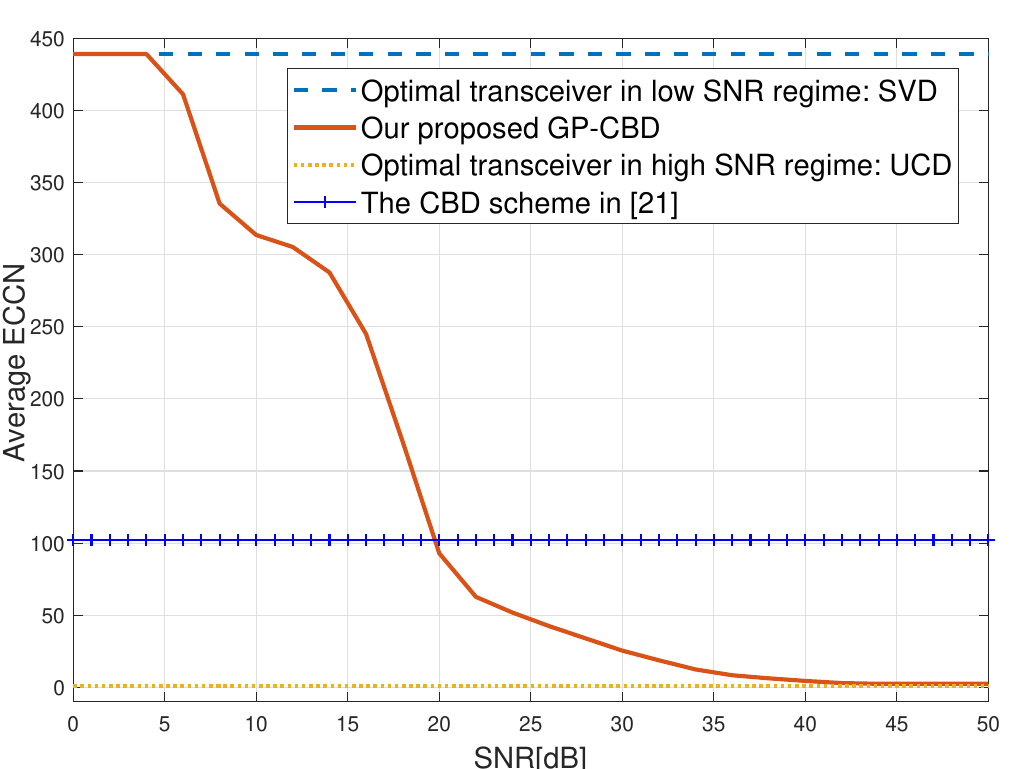}
        \subcaption{}
    \end{minipage}
    \begin{minipage}[b]{0.32\textwidth} 
        \centering
\includegraphics[width=2.25in, height=1.7in]{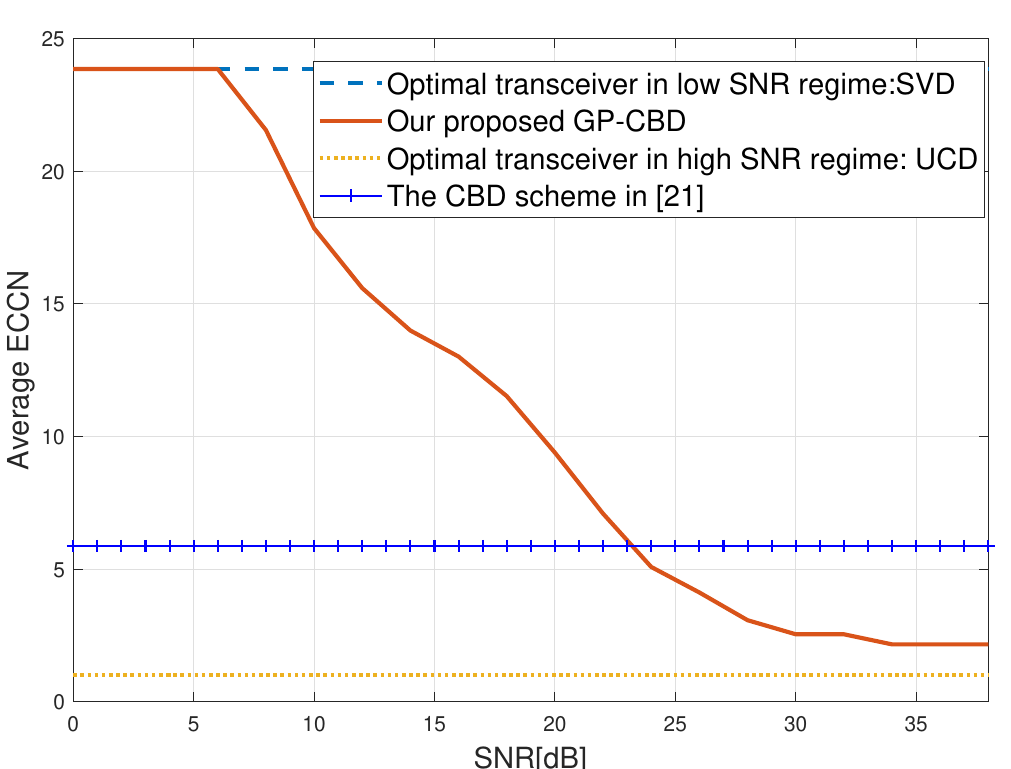}
        \subcaption{}
    \end{minipage}
\end{center}
    \caption{(a): The cumulative probability for the ECCN of different transceivers with SNR $=30$ dB. (b) and (c): The average ECCN of various transceivers vs. SNRs under CDL and Rayleigh channels, respectively.}
    \label{Fig:CN}
\end{figure*}

\subsection{ECCN Evaluation}
We generate the MIMO with i.i.d. Rayleigh fading channel \cite{liu2024leveraging} and CDL channel model specified in 3GPP TR 38.901 \cite{38901}. Each entry of the Rayleigh fading channel is modeled as an i.i.d. complex Gaussian variable $\mathcal{CN}(0,1)$. The CDL channel is generated based on the \textit{nrCDLChannel} function in MATLAB${ ^\circledR}$. The carrier frequency is set to $3.5$ GHz and the channel delay spread is $30$ ns. The antenna array is vertical polarized. The delays and angles of arrival and departure for the multipaths within a cluster are generally similar, which leads to a high spatial correlation in the MIMO channel. This high spatial correlation results in channel ill-conditioning, which is a characteristic challenge in massive MIMO systems.  

Fig. \ref{Fig:CN} evaluates the ECCN performance of different transceiver designs. Specifically, Fig. \ref{Fig:CN}(a) illustrates the cumulative distribution function of condition numbers under the $8\times8$ Rayleigh and CDL channels. The proposed GP-CBD method effectively reduces the ECCN compared to the SVD method at the SNR $=30$ dB. Figs. \ref{Fig:CN}(b) and \ref{Fig:CN}(c) show the average ECCN of different transceivers across the all SNR intervals. It can be observed that, with the same number of antennas, the CDL channel is more ill-conditioned than the Rayleigh channel. Since the UCD, SVD, and CBD methods are operated by a fixed-mode channel decomposition, they possess the fixed ECCNs. While the GP-CBD method efficiently adapts its ECCN across varying SNR levels. In the next subsection, we further show the superior performance of the proposed GP-CBD scheme in terms of the achievable rate.

\begin{figure*}[ht!]
\begin{center}
    \begin{subfigure}[b]{0.32\textwidth}
        \centering
        
        \includegraphics[width=2.25in, height=1.7in]{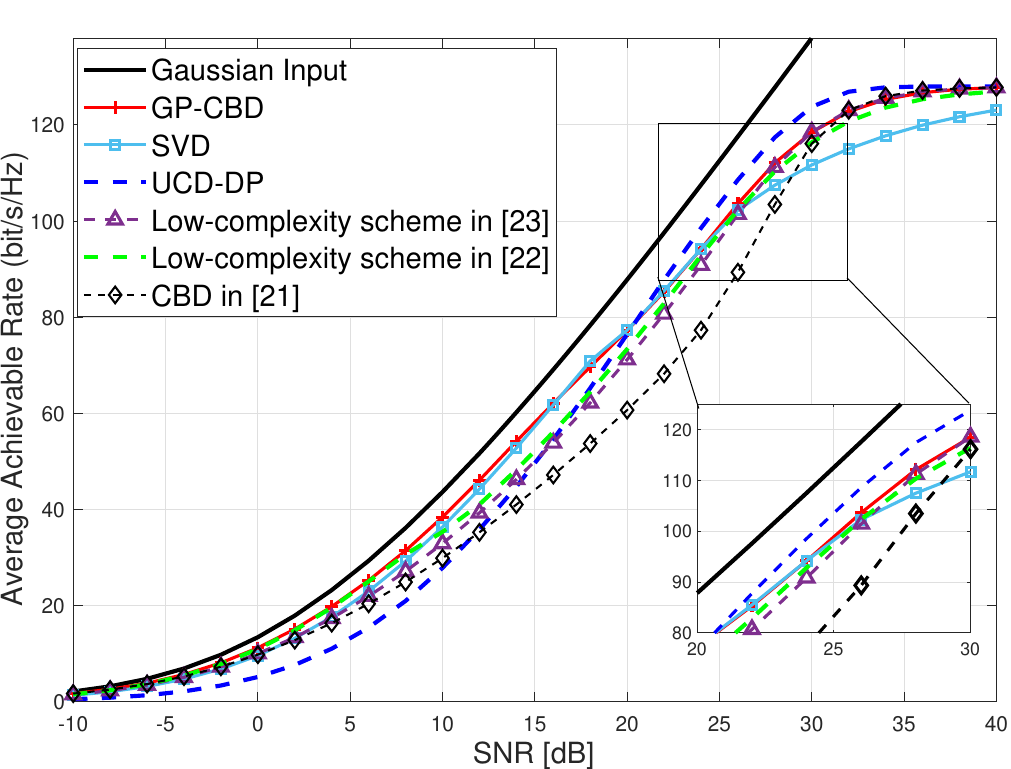}
        \subcaption{{$16\times 16$ MIMO-BICM channel with 256-QAM transmission.}} 
    \label{Fig:16×16_256QAM_MI}    
    \end{subfigure}
    \begin{subfigure}[b]{0.32\textwidth} 
        \centering
        \includegraphics[width=2.25in, height=1.7in]{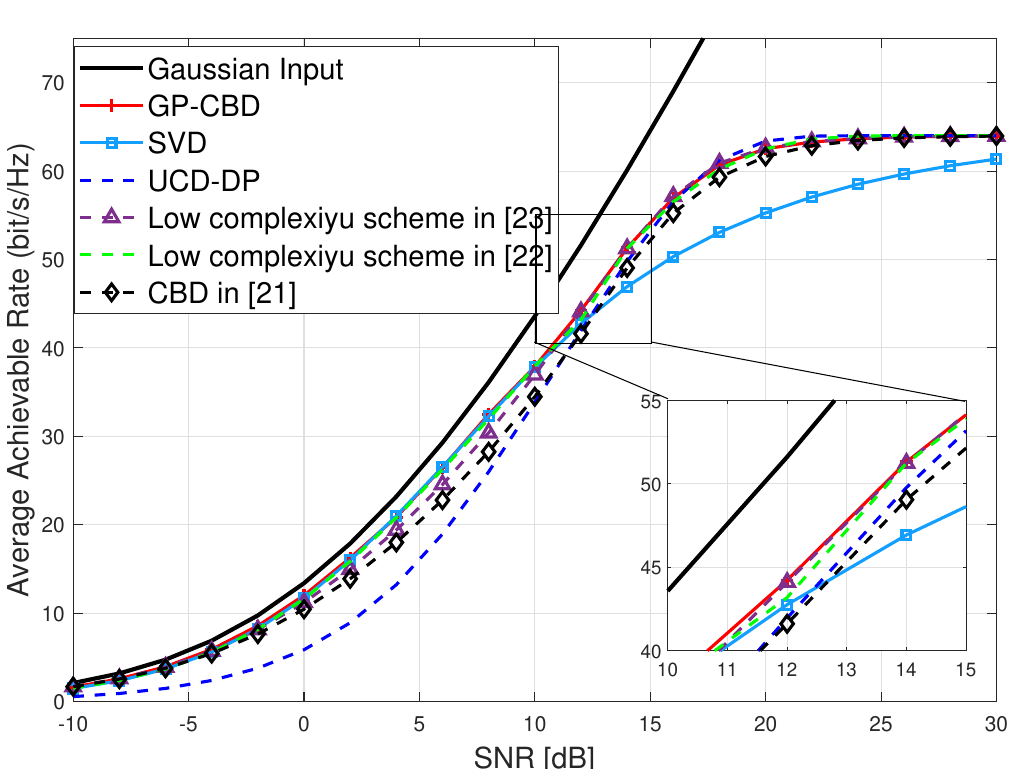}
        \subcaption{{$16\times16$ MIMO-BICM with 16-QAM transmission.}}
        \label{Fig:16×16_16QAM_MI}
    \end{subfigure}
    \begin{subfigure}[b]{0.32\textwidth} 
        \centering
\includegraphics[width=2.25in, height=1.7in]{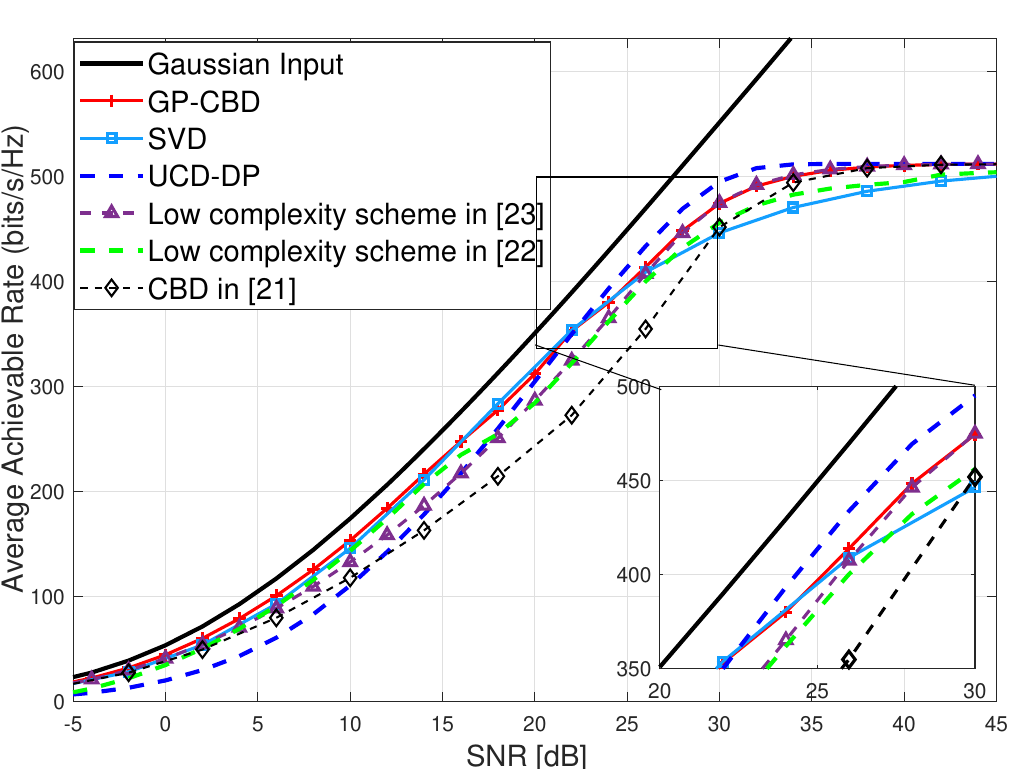}
        \subcaption{$64\times64$ MIMO-BICM with 256-QAM transmission.}
        \label{Fig:64×64_256QAM_MI}
    \end{subfigure}
\end{center}
    \caption{The achievable rate of the MIMO-BICM system under Rayleigh channels.}
    \label{Fig:1}
\end{figure*}

\begin{figure*}[!ht]
\begin{center}
    \begin{subfigure}[b]{0.32\textwidth} 
        \centering
        \includegraphics[width=2.25in, height=1.7in]{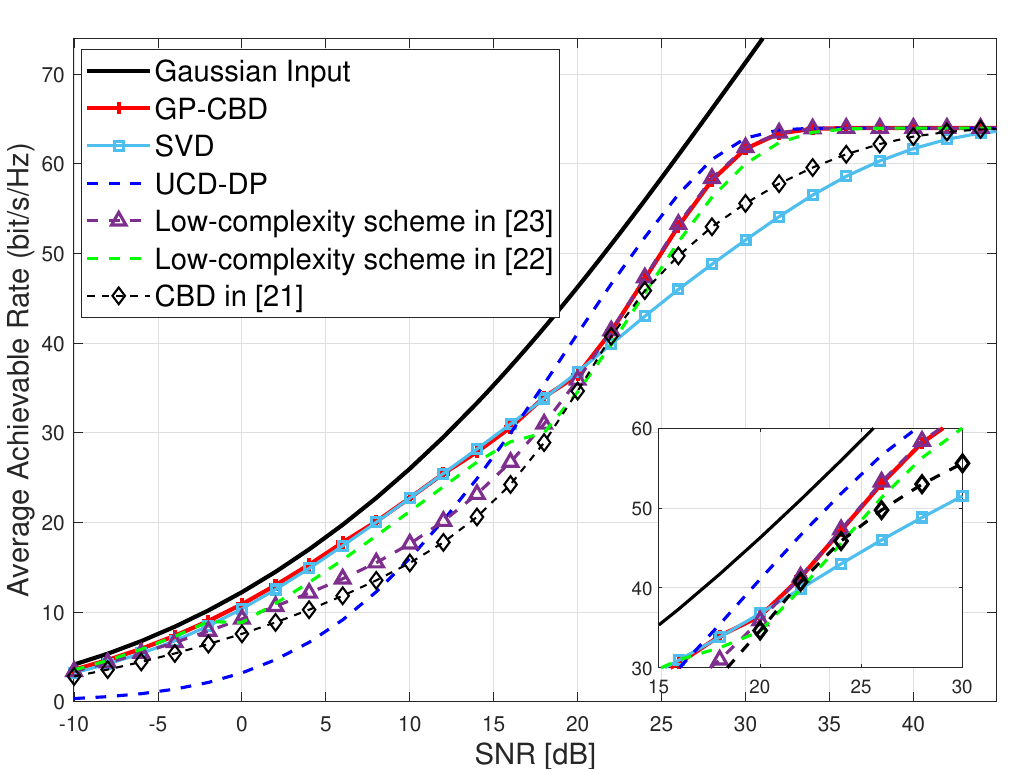}
        \subcaption{$8\times 256$ MIMO-BICM system  with 256-QAM transmission.} 
        \label{Fig:8×256_256QAM_CDL}
    \end{subfigure}
    \begin{subfigure}[b]{0.32\textwidth} 
        \centering
        \includegraphics[width=2.25in, height=1.7in]{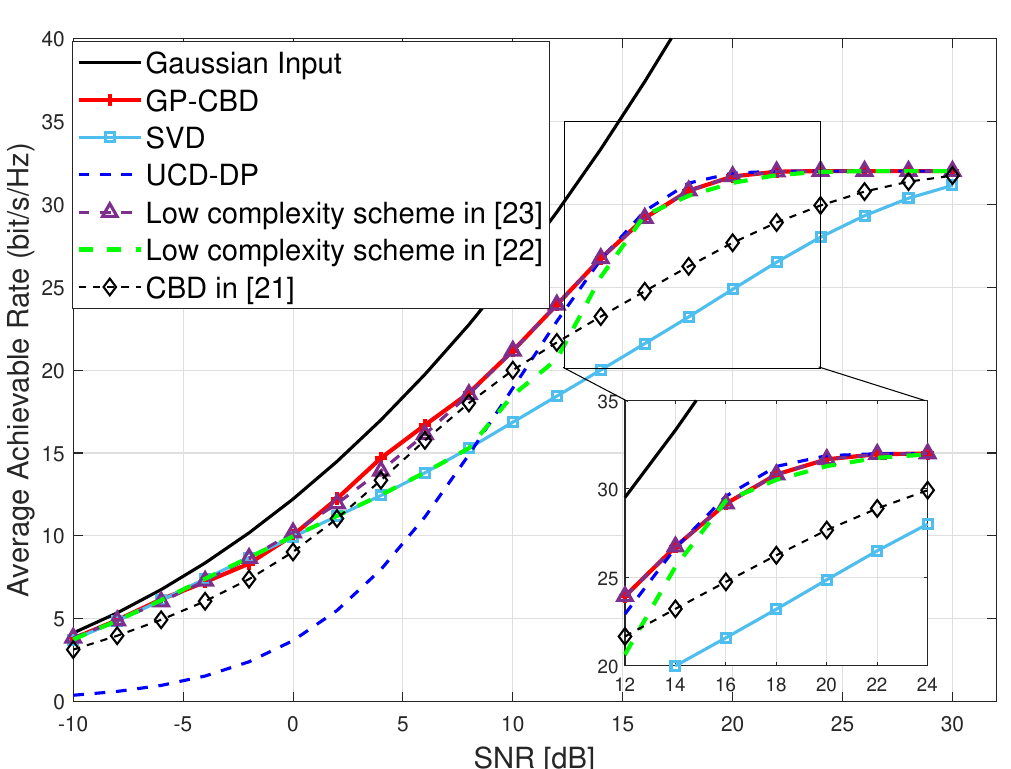}
        \subcaption{$8\times256$ MIMO-BICM system  with 16-QAM transmission.}
        \label{Fig:8×256_16QAM_MI_CDL}
    \end{subfigure}
    \begin{subfigure}[b]{0.32\textwidth} 
        \centering
\includegraphics[width=2.25in, height=1.7in]{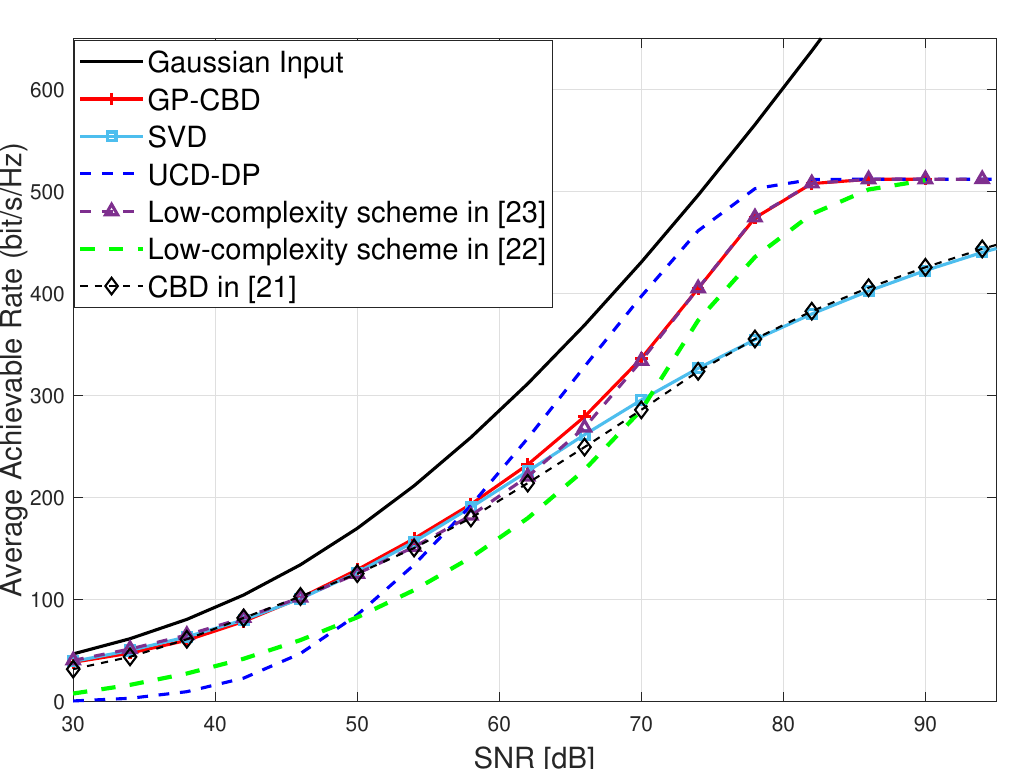}
        \subcaption{$64\times64$ MIMO-BICM system  with 256-QAM transmission.}
        \label{Fig:64×64_256QAM_CDL}
    \end{subfigure}
\end{center}
    \caption{The achievable rate of the massive MIMO-BICM system under CDL channels.}
    \label{Fig:2}
\end{figure*}

\subsection{Achievable Rate Performance}

{In this case, we characterize the system performance via ergodic achievable rate analysis, wherein GP-CBD is applied to individual channel realizations and the achievable rates are statistically averaged.} The ergodic achievable rate performance is evaluated in bits/s/Hz versus the SNR defined as $\rho \triangleq 10\log_{10}\frac{1}{\sigma^2_z}$ [dB], and the corresponding results are obtained by averaging 1000 independent trials. Figs. \ref{Fig:1} and \ref{Fig:2} exhibit the achievable rate under the Rayleigh and CDL channels, respectively.
For the CDL channel in Figs. \ref{Fig:2}(a) and  \ref{Fig:2}(b), the transmitter and receiver are equipped with a $16\times8$ and a $2\times2$ uniform plane array with vertical polarization, respectively. The antenna array in Fig. \ref{Fig:2}(c) is $8\times4$ uniform plane array with vertical polarization at both the transmitter and receiver. In the two cases, the \textit{log-max} SD has an unaffordable complexity; hence, we omit the performance of the \textit{log-max} SD here. 

In Figs. \ref{Fig:1} and \ref{Fig:2}, a visible fact is that the channel capacity under Gaussian inputs can grow unbounded as the SNR increases. With finite-alphabet inputs e.g., 256-QAM and 16-QAM, the achievable rate saturates to the peak value in high SNRs, and a reduced data rate (resulted from lower channel-coding rate) is supported to ensure reliable transmission in the low and medium SNR regimes. According to the previous analysis in Section \ref{SEC:III}, the rate performance of the SVD-MMSE is optimal with a low SNR, but is limited by some poorer eigen-subchannels when SNR is higher. In contrast, the rate performance of the UCD scheme is optimal in the high SNR regime, but it degrades at low SNRs. This is because the UCD decomposes the MIMO channel into multiple identical parallel eigen-subchannels with ECCN being $1$. 

In Fig. \ref{Fig:1}, the proposed GP-CBD scheme for the Rayleigh channel exhibits the  performance identical to the SVD-MMSE with low SNR. While in the high SNR scenario, GP-CBD significantly outperforms the SVD-MMSE and these low-complexity schemes in \cite{CBDbib,maleki2024precoding} (more than $5$ dB), and its performance closely approaches that of the UCD. Fig. \ref{Fig:2} shows the rate performance under the highly correlated CDL channel. In this case, our proposed GP-CBD has a more significant gain (exceeds than $10$ dB) compared to the SVD-based transceiver. Note that while our GP-CBD exhibits a slight performance degradation compared to the optimal UCD scheme in the high-SNR regime, this minor loss is compensated by its advantage of low complexity.

The impact of the varying number of antennas and modulation orders is also evaluated. With a fixed modulation order, Figs. \ref{Fig:1}(a) and (c) show a similar performance comparison among different algorithms with the varying antennas. While with a fixed number of antennas, Figs. \ref{Fig:2}(a) and (b) show that the higher modulation order transmission e.g. 256-QAM, needs higher SNR to saturate to the peak achievable rate than 16-QAM in the $8\times 256$ MIMO-BICM system. Specifically, for the 256-QAM and 16-QAM transmission, the achievable rate saturates at $64$ bits/s/Hz and $32$ bits/s/Hz at $30$ dB and $20$ dB, respectively. We also observe that there exists some crossover in Figs. (\ref{Fig:1}) and (\ref{Fig:2}). This phenomenon is consistent with the analysis conducted in Section \ref{SEC:III}. As part of our future work, we aim to develop a generalized transceiver design that achieves optimal achievable rate performance across the entire SNR range.

We also evaluate different power allocation strategies, as illustrated in Fig. \ref{Fig:waterfilling}, including classic water-filling (WF), mercury water-filling (MWF) \cite{1650354}, and average power allocation. The results show that MWF effectively improves the achievable rate under finite-alphabet inputs, while classic WF enhances performance in the low-SNR regime. Our proposed scheme, when combined with power allocation, achieves a superior achievable rate by improving performance at low SNRs and maintaining robustness at high SNRs for different power allocation schemes.


\subsection{BER Performance}
In this subsection, we evaluate the BER performance of different transceivers. We adopt an LDPC code \cite{3GPP} as the channel coding scheme, and perform transmissions in accordance with the modulation and coding scheme (MCS) table described in \cite[Table 5.1.3.1-2: MCS index table 2 for PDSCH]{38214}. In the simulations, we set that each frame consists of $4992\times Q_m$ bits, { and is transmitted in quasi-static block fading channels. These channels are i.i.d over the time.}



 
\begin{figure}[tb]
\centering
\includegraphics[width=3in]{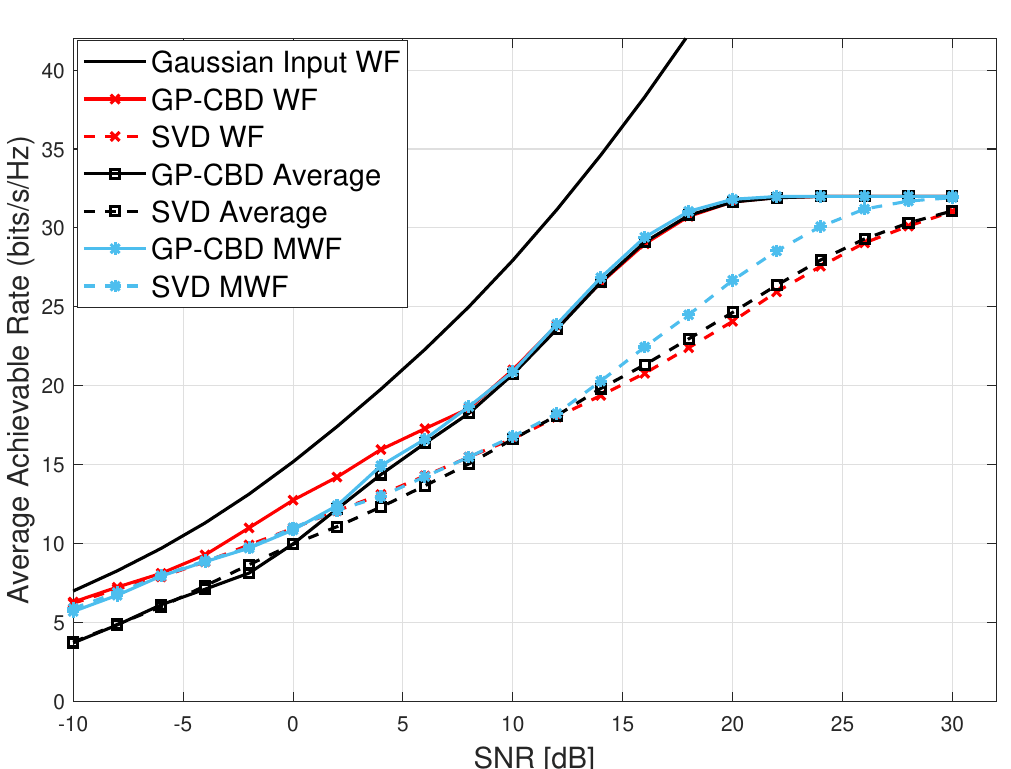}
\caption{The achievable rate performance of different power allocation methods for 16-QAM transmission under $8\times 256$ CDL channels.}
\label{Fig:waterfilling}
\end{figure}

\begin{figure}[tb]
\centering
\includegraphics[width=3.4in]{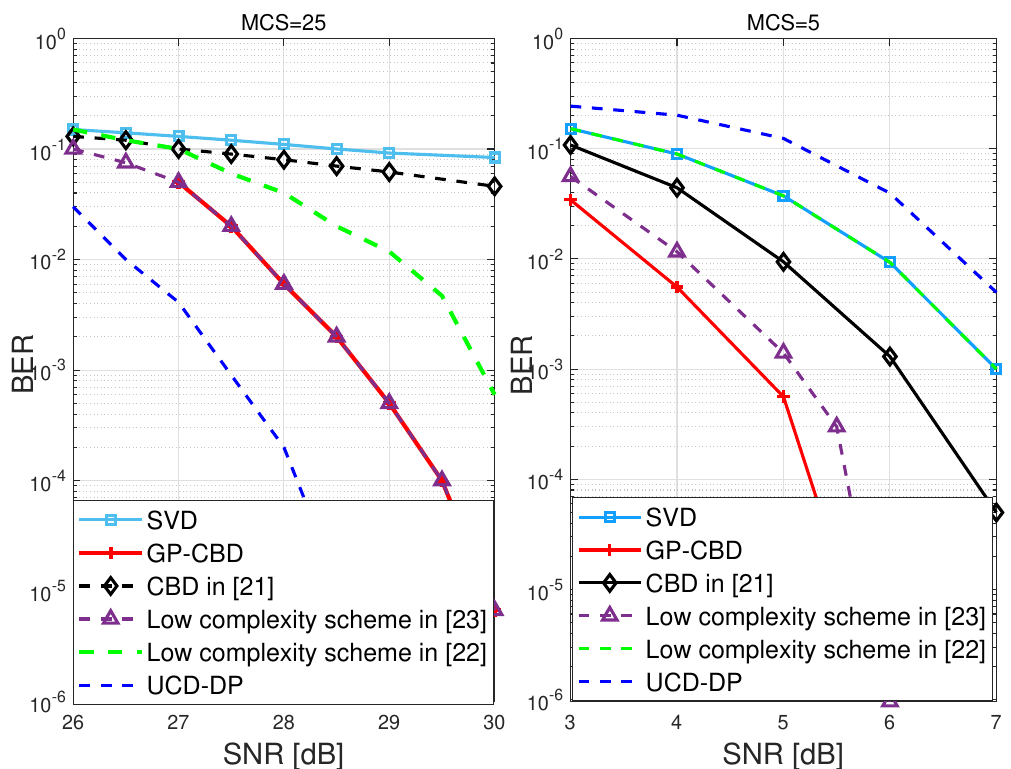}
\caption{{BER performance of the $8\times256$ coded MIMO-BICM systems with MCS25 and MCS5 transmissions under the CDL channel.}}
\label{Fig:mcs25_mcs5_8×256}
\end{figure}

\begin{figure}[htb]
\centering
\includegraphics[width=3.4in]{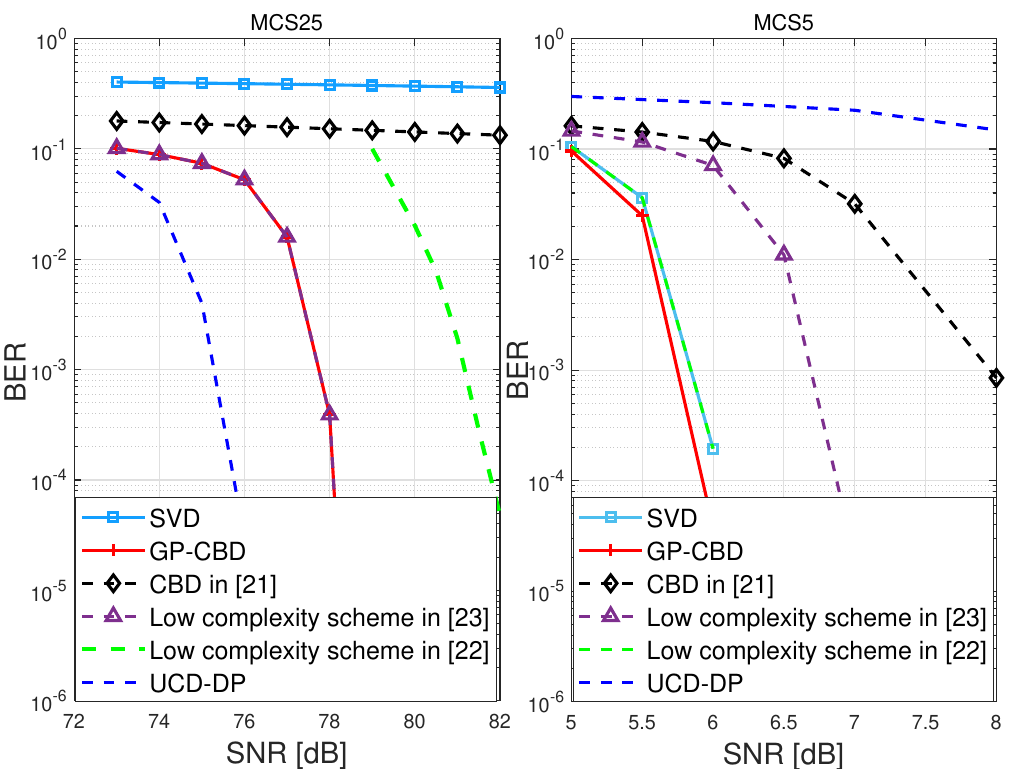}
\caption{{BER performance for a $64\times64$ coded MIMO-BICM system with MCS25 and MCS5 transmissions under the CDL and Rayleigh channel, respectively.}}
\label{Fig:mcs25_mcs5_64×64}
\end{figure}

Fig. \ref{Fig:mcs25_mcs5_8×256} shows the BER performance comparison of the $8 \times 256$ coded MIMO-BICM system using 16-QAM (MCS5 with spectral efficiency (SE) of $11.82$ bits/s/Hz) and 256-QAM (MCS25 with SE of $55.3$ bits/s/Hz) transmissions under the CDL channel. In the case of MCS5, our GP-CBD clearly outperforms other algorithms, which is consistent with the results in Fig. \ref{Fig:2}(b). Specifically, at the BER level of $10^{-4}$, GP-CBD achieves a gain of approximately $0.35$ dB, $2$ dB, and $1.5$ dB over the low-complexity schemes in \cite{PL-CBD, maleki2024precoding}, and the CBD, respectively. For MCS25, the UCD-DP scheme achieves superior performance among these transmission strategies. In this case, the BER performance of the proposed GP-CBD design matches that of the scheme in \cite{PL-CBD}, and outperforms SVD-MMSE, CBD, and the low-complexity scheme in \cite{maleki2024precoding}. Note that the GP-CBD suffers about $1.2$ dB performance loss compared to the UCD-DP scheme.

Fig. \ref{Fig:mcs25_mcs5_64×64} shows the BER performance comparison for a $64 \times 64$ coded MIMO-BICM system with MCS25 (SE = $442.5$ bits/s/Hz) under the CDL channel and MCS5 (SE = $94.5$ bits/s/Hz, corresponding to $23.63$ bits/s/Hz in Fig. \ref{Fig:1}(b)) transmission under a Rayleigh fading channel, respectively. In the case of MCS25, it is clear that the BER performance of SVD-MMSE and CBD dramatically degrades due to the large ECCN. While the proposed GP-CBD can alleviate this degradation. In the case of MCS5, our proposed GP-CBD is reduced to the SVD, and has the superior BER and rate performance here. Notice that the SVD-MMSE has a slight loss due to the MMSE approximation operation for obtaining the LLR. In additional, the trend of numerical BER results are consistent with the achievable rate performance in different channel models. All simulation results show that the BER performance of our GP-CBD outperforms both the SVD-MMSE and the original CBD schemes, with improvements being more notable in high-SE scenarios. Moreover, the GP-CBD achieves BER performance comparable to the UCD-DP scheme in high-code-rate scenarios, while maintaining significantly lower complexity, especially in massive MIMO systems. This well justifies that our GP-CBD scheme could be a competitive transceiver design solution for future massive MIMO systems.




 

\section{Conclusion} \label{SEC:conclusion}
In this paper, we presented a theoretical analysis of the achievable rate for CBD transceivers in MIMO-BICM systems. We demonstrated that leveraging MI to evaluate achievable rate performance is both simple and efficient. Based on the derived MI analysis, we identified the desired CBD structure to maximize the achievable rate. Specifically, in low SNR regimes (corresponding to low channel coding rates), the optimal CBD reduces to SVD-based designs. In contrast, in high SNR regimes (corresponding to high channel coding rates), the optimal transceiver designs require a CBD structure with equal diagonal entries. Building on these insights, we proposed a GP-CBD scheme that supports low-complexity parallel implementation. Simulation results show that GP-CBD outperforms the original CBD design proposed in \cite{CBDbib} in terms of both achievable rate and BER performance. Moreover, GP-CBD demonstrates significant advantages over other existing state-of-the-art schemes, highlighting its potential for application in next-generation communications that demand high spectral efficiency and low latency.  Our analysis suggested that integrating the advantages of both SVD and UCD transceiver designs could form a generalized optimal channel decomposition that maximizes the achievable rate across the entire SNR range. This presents an intriguing direction for our future research.

\appendices
\section{Proof of Theorem 1}\label{App:A}
Using the Jensen inequality, we have
    \begin{equation}\label{eq:App41}
    \begin{aligned}
         \rm R_{ MIMO-BICM} & =  \sum^{K}_{k=1}\left( 1-\mathbb E_{c_k, \mathbf{z}} \left\{\log_2 \left(1+e^{(1-2c_k)L_k} \right) \right\} \right) \\
         & \geq \sum^{K}_{k=1}\left( 1- \log_2 \left(1+\mathbb E_{c_k, \mathbf{z}}\{e^{(1-2c_k)L_k}\} \right)  \right).
    \end{aligned}
    \end{equation}
Based on (\ref{eq:eqprobe}), we have 
\begin{equation}\label{eq:app1}
    \mathbb E_{c_k, \mathbf{z}}\{e^{(1-2c_k)L_k}\} = \frac{1}{2}\mathbb E_{\mathbf{z}}\{e^{-L_k(c_k=1)}\}+\frac{1}{2}\mathbb E_{\mathbf{z}}\{e^{L_k(c_k=0)}\},
\end{equation}
and
\begin{equation}
    \begin{aligned}
        &L_k(c_k=1) = \frac{1}{\sigma^2_z}\left( \mathop{\min}_{\mathbf{s}^1\in \mathcal{X}^{1},\mathbf{s}^0\in \mathcal{X}^{0}}\|\mathbf{H}(\mathbf{s}^1-\mathbf{s}^0)+\mathbf{z}\|^2- \|\mathbf{z}\|^2\right) \\
        &= \min \left( \frac{1}{\sigma^2_z} \left( \|\mathbf{H}(\mathbf{s}^1-\mathbf{s}^0)\|^2+2\Re\{\mathbf{z}^H\mathbf{H}(\mathbf{s}^1-\mathbf{s}^0)\}\right)\right),
        \\
        &L_k(c_k=0) = \frac{1}{\sigma^2_z}\left( \|\mathbf{z}\|^2- \mathop{\min}_{\mathbf{s}^1\in \mathcal{X}^{1},\mathbf{s}^0\in \mathcal{X}^{0}}\|\mathbf{H}(\mathbf{s}^0-\mathbf{s}^1)+\mathbf{z}\|^2\right)\\
        &=\min \left(\frac{1}{\sigma^2_z}\left( -\|\mathbf{H}(\mathbf{s}^0-\mathbf{s}^1)\|^2+2\Re\{\mathbf{z}^H\mathbf{H}(\mathbf{s}^1-\mathbf{s}^0)\}\right)\right).
    \end{aligned}
\end{equation}
For notational simplicity, we define $\mathbf{e} = \mathbf{s}^1-\mathbf{s}^0$. In high SNR regimes, with $\sigma_z^2 \rightarrow 0$, we easily obtain $L_k(c_k=1) \approx \min( \frac{\|\mathbf{He}\|^2}{\sigma^2_z})$ and $L_k(c_k=0)\approx \min(-\frac{\|\mathbf{He}\|^2}{\sigma^2_z} )$. Substituting them into (\ref{eq:app1}), we have 
\begin{equation}
    \mathbb E_{c_k, \mathbf{z}}\{e^{(1-2c_k)L_k}\} = e^{\frac{-\min \|\mathbf{He}\|^2}{\sigma_z^2}}.
\end{equation}
 
In low SNR regimes, with $\sigma_z^2 \rightarrow \infty$, it is clear that the signal is drowned in noise. Then, we have
\begin{equation}
\begin{aligned}
    \mathbb E_{c_k, \mathbf{z}}\{e^{(1-2c_k)L_k}\}  &= \frac{1}{2(\pi\sigma_z)^{N_r}}\int \left\{e^{\frac{-\min (\|\mathbf{He+z}\|^2)+\mathbf{z}^2}{\sigma_z^2}}e^{\frac{-\|\mathbf{z}\|^2}{\sigma_z^2}} \right. \\ &+ \left. e^{\frac{-\min (\|\mathbf{-He+z}\|^2)+\mathbf{z}^2}{\sigma_z^2}}e^{\frac{-\|\mathbf{z}\|^2}{\sigma_z^2}} \right\} d\mathbf{z} \\
     \approx \frac{1}{(\pi\sigma_z)^{N_r}}&\int e^{\frac{-\min((\mathbf{He})^H\mathbf{He}+\|\mathbf{z}\|^2)}{\sigma^2_z}}d\mathbf{z} = e^{\frac{-\min((\mathbf{He})^H\mathbf{He})}{\sigma^2_z}}.
\end{aligned}
\end{equation}
Moreover, $\mathbf{e}^H\mathbf{e}\geq e_i^He_i\geq d^2_{\rm min}$ with $e_i \neq 0$ and $d^2_{\rm min}$ being the MED of the constellation set $\mathcal{M}$. Hence, we have 
\begin{equation}
\min((\mathbf{He})^H\mathbf{He})\geq d^2_{\rm min}\mathbf{h}_i^H\mathbf{h}_i,
\end{equation}
where $\mathbf{h}_i$ is the $i$-th column of $\mathbf{H}$. Substitute it into (\ref{eq:App41}). Then we obtain
\begin{equation}
    {\rm R_{MIMO-BCIM}} \geq \sum^{N_s}_{i=1}\sum^{Q_m}_{k=1}\left( 1- \log_2 \left(1+e^{-\frac{d^2_{\rm min}\mathbf{h}_i^H\mathbf{h}_i}{\sigma_z^2}} \right)  \right),
\end{equation}
which concludes the proof of Theorem 1.

\bibliographystyle{ieeetr}
\bibliography{all}

\begin{thebibliography}{10}

\bibitem{BICM}
G.~Caire, G.~Taricco, and E.~Biglieri, ``{Bit-interleaved coded modulation},'' {\em IEEE Trans. Inf. Theory}, vol.~44, no.~3, pp.~927--946, 1998.

\bibitem{wang2020joint}
P.~Wang, J.~Fang, L.~Dai, and H.~Li, ``Joint transceiver and large intelligent surface design for massive {MIMO} mmwave systems,'' {\em IEEE Trans. Wireless Commun.}, vol.~20, no.~2, pp.~1052--1064, 2020.

\bibitem{bjornson2017massive}
E.~Bj{\"o}rnson, J.~Hoydis, and L.~Sanguinetti, ``Massive {MIMO} has unlimited capacity,'' {\em IEEE Trans. Wireless Commun.}, vol.~17, no.~1, pp.~574--590, 2017.

\bibitem{9343768}
N.~S. Perović, L.-N. Tran, M.~Di~Renzo, and M.~F. Flanagan, ``Achievable rate optimization for {MIMO} systems with reconfigurable intelligent surfaces,'' {\em IEEE Trans. Wireless Commun.}, vol.~20, no.~6, pp.~3865--3882, 2021.

\bibitem{jing2021linear}
S.~Jing and C.~Xiao, ``{Linear MIMO precoders with finite alphabet inputs via stochastic optimization and deep neural networks (DNNs)},'' {\em IEEE Trans. Signal Process.}, vol.~69, pp.~4269--4281, 2021.

\bibitem{9693344}
G.~Xia, X.~Zhou, L.~Gu, F.~Shu, Y.~Wu, and J.~Wang, ``Joint precoder and beamformer design for secure relay networks with finite-alphabet inputs and statistical {CSI} of {Eve},'' {\em IEEE Trans. Wireless Commun.}, vol.~21, no.~8, pp.~5814--5827, 2022.

\bibitem{10534211}
A.~Papazafeiropoulos, J.~An, P.~Kourtessis, T.~Ratnarajah, and S.~Chatzinotas, ``Achievable rate optimization for stacked intelligent metasurface-assisted holographic {MIMO} communications,'' {\em IEEE Trans. Wireless Commun.}, vol.~23, no.~10, pp.~13173--13186, 2024.

\bibitem{SD}
B.~Hochwald and S.~Brink, ``{Achieving near-capacity on a multiple-antenna channel},'' {\em IEEE Trans. Commun.}, vol.~51, no.~3, pp.~389--399, 2003.

\bibitem{barbero2008fixing}
L.~G. Barbero and J.~S. Thompson, ``Fixing the complexity of the sphere decoder for {MIMO} detection,'' {\em IEEE Trans. Wireless Commun.}, vol.~7, no.~6, pp.~2131--2142, 2008.

\bibitem{takahashi2021low}
T.~Takahashi, A.~T{\"o}lli, S.~Ibi, and S.~Sampei, ``{Low-complexity large MIMO detection via layered belief propagation in beam domain},'' {\em IEEE Trans. Wireless Commun.}, vol.~21, no.~1, pp.~234--249, 2021.

\bibitem{10480363}
T.~Takahashi, H.~Iimori, K.~Ishibashi, S.~Ibi, and G.~T.~F. de~Abreu, ``Bayesian bilinear inference for joint channel tracking and data detection in millimeter-wave {MIMO} systems,'' {\em IEEE Trans. Wireless Commun.}, vol.~23, no.~9, pp.~11136--11153, 2024.

\bibitem{zheng2025low}
J.~Zheng, Y.~Sun, H.~Zhou, W.~Zhou, Y.~Huang, X.~You, and C.~Zhang, ``Low-complexity breadth-first search detection for large-scale {MIMO} systems,'' {\em IEEE Trans. Commun.}, 2025.

\bibitem{OLTD}
J.~Yang, Q.~Du, and Y.~Jiang, ``Neural network-assisted receiver design via learning trellis diagram online,'' {\em IEEE Trans. Commun}, vol.~70, no.~12, pp.~8075--8085, 2022.

\bibitem{PMAP}
L.~Bai, Q.~Zeng, R.~Han, J.~Choi, and W.~Zhang, ``Deep learning-based low complexity {MIMO} detection via partial {MAP},'' {\em IEEE Trans. Wireless Commun.}, vol.~24, no.~3, pp.~2126--2139, 2025.

\bibitem{telatar1999capacity}
E.~Telatar, ``{Capacity of multi-antenna Gaussian channels},'' {\em European Trans. Telecommun.}, vol.~10, no.~6, pp.~585--595, 1999.

\bibitem{perovic2021achievable}
N.~S. Perovi{\'c}, L.-N. Tran, M.~Di~Renzo, and M.~F. Flanagan, ``{Achievable rate optimization for MIMO systems with reconfigurable intelligent surfaces},'' {\em IEEE Trans. Wireless Commun.}, vol.~20, no.~6, pp.~3865--3882, 2021.

\bibitem{GMD}
Y.~Jiang, J.~Li, and W.~Hager, ``{Joint transceiver design for {MIMO} communications using geometric mean decomposition},'' {\em IEEE Trans. Signal Process.}, vol.~53, no.~10, pp.~3791--3803, 2005.

\bibitem{UCD}
Y.~Jiang, J.~Li, and W.~Hager, ``{Uniform channel decomposition for {MIMO} communications},'' {\em IEEE Trans. Signal Process.}, vol.~53, no.~11, pp.~4283--4294, 2005.

\bibitem{mohammed2011mimo}
S.~K. Mohammed, E.~Viterbo, Y.~Hong, and A.~Chockalingam, ``{MIMO precoding with X-and Y-codes},'' {\em IEEE Trans. Inf. Theory}, vol.~57, no.~6, pp.~3542--3566, 2011.

\bibitem{mohammed2011precoding}
S.~K. Mohammed, E.~Viterbo, Y.~Hong, and A.~Chockalingam, ``{Precoding by pairing subchannels to increase MIMO capacity with discrete input alphabets},'' {\em IEEE Trans. Inf. Theory}, vol.~57, no.~7, pp.~4156--4169, 2011.

\bibitem{CBDbib}
J.~Yang, W.~Hu, and Y.~Jiang, ``{Channel} bi-diagonalization for {MIMO} communications,'' {\em IEEE Wireless Commun. Lett.}, vol.~12, no.~1, pp.~163--167, 2023.

\bibitem{maleki2024precoding}
M.~Maleki, J.~Jin, H.~Wang, and M.~Haardt, ``Precoding design and {PMI} selection for {BICM-MIMO} systems with {5G} new radio type-{I} {CSI},'' {\em IEEE Trans. Commun.}, 2024.

\bibitem{PL-CBD}
J.~Yang, W.~Hu, Y.~Jiang, and X.~Wang, ``Parallel channel bidiagonalization for massive {MIMO }systems,'' in {\em Proc. IEEE/CIC ICCC}, pp.~717--722, 2024.

\bibitem{fertl2011performance}
P.~Fertl, J.~Jalden, and G.~Matz, ``{Performance assessment of {MIMO}-BICM demodulators based on mutual information},'' {\em IEEE Trans. Signal Process.}, vol.~60, no.~3, pp.~1366--1382, 2011.

\bibitem{38901}
{3GPP TS 38.901}, ``{Study on channel model for frequencies from 0.5 to 100 GHz}.''
\newblock 3rd Generation Partnership Project; Technical Specification Group Radio Access Network, 2020.

\bibitem{SoftSD}
C.~Studer, A.~Burg, and H.~Bolcskei, ``{Soft-output sphere decoding: Algorithms and VLSI implementation},'' {\em IEEE J. Sel. Areas Commun.}, vol.~26, no.~2, pp.~290--300, 2008.

\bibitem{SOftSD1}
J.~Jalden and B.~Ottersten, ``{Parallel} implementation of a soft output sphere decoder,'' in {\em Proc. Asilomar, 2005.}, pp.~581--585, 2005.

\bibitem{EqualQR}
J.-K. Zhang, A.~Kavcic, and K.~M. Wong, ``{Equal-diagonal QR decomposition and its application to precoder design for successive-cancellation detection},'' {\em IEEE Trans. Inf. Theory}, vol.~51, no.~1, pp.~154--172, 2005.

\bibitem{horn}
A.~Horn, ``{On the eigenvalues of a matrix with prescribed singular values},'' {\em Proc. Amer. Math. Soc.}, vol.~5, no.~1, pp.~4--7, 1954.

\bibitem{liu2024leveraging}
J.~Liu, Y.~Ma, J.~Wang, and R.~Tafazolli, ``Accelerating iteratively linear detectors in multi-user {(ELAA-)MIMO} systems with {UW-SVD},'' {\em IEEE Trans. Wireless Commun}, 2024.

\bibitem{9133524}
S.~Sedighi and E.~Ayanoglu, ``Bit-interleaved coded multiple beamforming in millimeter-wave massive {MIMO} systems,'' {\em IEEE Trans. Commun}, vol.~68, no.~10, pp.~6174--6185, 2020.

\bibitem{tse2005fundamentals}
D.~Tse and P.~Viswanath, {\em {Fundamentals of wireless communication}}.
\newblock Cambridge University Press, 2005.

\bibitem{viterbi1971convolutional}
A.~Viterbi, ``{Convolutional codes and their performance in communication systems},'' {\em IEEE Trans. Commun. Technol.}, vol.~19, no.~5, pp.~751--772, 1971.

\bibitem{BCJR}
L.~Bahl, J.~Cocke, F.~Jelinek, and J.~Raviv, ``{Optimal decoding of linear codes for minimizing symbol error rate (corresp.)},'' {\em {IEEE Trans. Inf. Theory}}, vol.~20, no.~2, pp.~284--287, 1974.

\bibitem{BICM-ID}
A.~Chindapol and J.~Ritcey, ``{Design, analysis, and performance evaluation for BICM-ID with square QAM constellations in Rayleigh fading channels},'' {\em IEEE J. Sel. Areas Commun.}, vol.~19, no.~5, pp.~944--957, 2001.

\bibitem{biglieri2000bit}
E.~Biglieri, G.~Taricco, and E.~Viterbo, ``{Bit-interleaved time-space codes for fading channels},'' 2000.

\bibitem{mckay2005capacity}
M.~R. McKay and I.~B. Collings, ``{Capacity and performance of {MIMO}-BICM with zero-forcing receivers},'' {\em IEEE Trans. Commun.}, vol.~53, no.~1, pp.~74--83, 2005.

\bibitem{jiang2005geometric}
Y.~Jiang, W.~W. Hager, and J.~Li, ``{The geometric mean decomposition},'' {\em Linear Algebra Appl.}, vol.~396, pp.~373--384, 2005.

\bibitem{golub2013matrix}
G.~H. Golub and C.~F. Van~Loan, {\em {Matrix Computations}}.
\newblock Johns Hopkins University Press, 2013.

\bibitem{1268365}
S.~Liu and Z.~Tian, ``{Near-optimum soft decision equalization for frequency selective MIMO channels},'' {\em IEEE Trans. Signal Process.}, vol.~52, no.~3, pp.~721--733, 2004.

\bibitem{palomardaniel2006MIMO}
D.~P. Palomar and Y.~Jiang, ``{{MIMO}} transceiver design via majorization theory,'' {\em Found. Trends Commun. Inf. Theory}, 2006.

\bibitem{olmos2012use}
J.~Olmos, A.~Serra, S.~Ruiz, and I.~Latif, ``{On the use of mutual information at bit level for accurate link abstraction in LTE with incremental redundancy H-ARQ},'' {\em COST IC1004 TD (12)}, vol.~5046, pp.~1--21, 2012.

\bibitem{1650354}
A.~Lozano, A.~Tulino, and S.~Verdu, ``Optimum power allocation for parallel {Gaussian} channels with arbitrary input distributions,'' {\em IEEE Trans. Inf. Theory}, vol.~52, no.~7, pp.~3033--3051, 2006.

\bibitem{3GPP}
{3GPP TS 38.212}, ``{NR; Multiplexing and channel coding}.''
\newblock 3rd Generation Partnership Project; Technical Specification Group Radio Access Network, 2019.

\bibitem{38214}
{3GPP TS 38.214}, ``{NR; Physical layer procedures for data}.''
\newblock 3rd Generation Partnership Project; Technical Specification Group Radio Access Network, 2020.

\end{thebibliography}
\end{document}